%% file: main.tex
\documentclass[lettersize,journal]{IEEEtran}

\IEEEoverridecommandlockouts
\usepackage{cite}
\usepackage{amsmath,amssymb,amsfonts}
\usepackage{subfig}
\usepackage{graphicx}
\usepackage{textcomp}
\usepackage{xcolor}
\usepackage{xfrac}
\usepackage{hyperref}
\usepackage[utf8]{inputenc} 
\usepackage[T1]{fontenc}
\usepackage{url}
\usepackage{ifthen}
\usepackage{cite}
\usepackage{graphicx}
\usepackage{amsfonts}
\usepackage{amssymb}
\usepackage{amsthm}
\usepackage{bbm}
\usepackage{diagbox}
\usepackage{tikz}
\usepackage{booktabs,multirow}
\usepackage{array}
\usepackage{caption}

\usepackage{float}
\usepackage[symbol]{footmisc}

\usepackage{authblk}

\usepackage[ruled,vlined]{algorithm2e}

\usetikzlibrary{decorations.markings}

\usepackage{prettyref,xspace}

\newtheorem{theorem}{Theorem}
\newtheorem{remark}{Remark}

\tikzset{->-/.style={decoration={
  markings,
  mark=at position #1 with {\pgftransformscale{1.5}\arrow{>}}},postaction={decorate}}}
 
 \tikzset{-<-/.style={decoration={
  markings,
  mark=at position #1 with {\pgftransformscale{1.5}\arrow{<}}},postaction={decorate}}}

\def\BibTeX{{\rm B\kern-.05em{\sc i\kern-.025em b}\kern-.08em
    T\kern-.1667em\lower.7ex\hbox{E}\kern-.125emX}}

\newcommand{\figref}[1]{Fig.~\ref{#1}}  
\newcommand{\tabref}[1]{Table~\ref{#1}}
\newcommand{\algref}[1]{Algorithm~\ref{#1}}

\newcommand{\secref}[1]{Section~\ref{#1}}

\newcommand{\appref}[1]{Appendix~\ref{#1}}

\newcommand{\lightcode}{\textsc{LightCode}\xspace}
\newcommand{\deeppolar}{\textsc{DeepPolar}\xspace}
\newcommand{\gn}{Gallager-Nakibo\u{g}lu\xspace}
\newcommand{\pb}{\textsc{PowerBlast}\xspace}

\DeclareMathOperator{\E}{E}

\def\E{\mathbb{E}}

\def\Pr{\mathbb{P}}

\begin{document}

\renewcommand{\thealgocf}{\arabic{algocf}}


\title{\lightcode: Light Analytical and Neural Codes for Channels with Feedback

}


\author[1]{Sravan Kumar Ankireddy\thanks{Correspondence to: Sravan Ankireddy \& Hyeji Kim (email: \\ \{sravan.ankireddy,hyeji.kim\}@utexas.edu, \\ {Source code available at: \href{https://github.com/sravan-ankireddy/lightcode}
{https://github.com/sravan-ankireddy/lightcode}}}}
\author[2]{Krishna Narayanan}
\author[1]{Hyeji Kim}
\affil[1]{University of Texas at Austin} 
\affil[2]{Texas A\&M University, College Station}

\renewcommand\Authands{, }


\maketitle

\input{abstract}

\input{intro}
\input{sysmodel}


\input{analytical}

\input{powerblast}

\input{neural_fb_codes}

\input{algorithm}

\input{results}

\input{analysis}
\input{noisy}
\input{conclusion}

\input{bibilography}


\input{appendix}








\end{document}

%% file: abstract.tex
\begin{abstract}\label{sec:abstract}
    The design of reliable and efficient codes for channels with feedback remains a longstanding challenge in communication theory. 
    While significant improvements have been achieved by leveraging deep learning techniques, neural codes often suffer from high computational costs, a lack of interpretability, and limited practicality in resource-constrained settings. 
%
    %
    We focus on designing low-complexity coding schemes that are interpretable and more suitable for communication systems. We advance both analytical and neural codes. First, we demonstrate that \pb, an analytical coding scheme inspired by Schalkwijk-Kailath (SK) and \gn (GN) schemes, achieves notable reliability improvements over both SK and GN schemes, outperforming neural codes in high signal-to-noise ratio (SNR) regions.
    Next, to enhance reliability in low-SNR regions, 
    we propose \lightcode, a lightweight neural code that achieves state-of-the-art reliability while using a fraction of memory and compute compared to existing deep-learning-based codes. Finally, we systematically analyze the learned codes, establishing connections between \lightcode and \pb, identifying components crucial for performance, and providing interpretation aided by linear regression analysis.
    
    
\end{abstract}

\begin{IEEEkeywords}
Channels with Feedback, Deep Learning, Channel Coding, Feedback Coding, Finite Block Length Coding
\end{IEEEkeywords}


%% file: intro.tex
\section{Introduction}\label{sec:intro}

Shannon introduced the concept of feedback channel in~\cite{shannon1956zero}. In a channel with feedback, the transmitter \textit{cooperates} with the receiver to improve the probability of successful transmission by utilizing the \textit{feedback} from the receiver. 
In~\cite{shannon1956zero}, Shannon assumes a perfect {\em noiseless} feedback channel with {\em unit delay} and demonstrates that the availability of feedback at the transmitter does not change the capacity of the resultant forward channel for memoryless channels. Interestingly, while the capacity remains the same, significant improvements in error exponents can be achieved with the help of feedback. 


One of the seminal schemes in the noiseless feedback setting has been provided by Shalkwijk and Kailath in~\cite{schalkwijk1966coding_p1,schalkwijk1966coding_p2} (SK) using a simple linear encoding scheme resulting in a doubly exponential error exponent for finite block lengths. 
In~\cite{gallager2009variations}, a two-phase variant of the SK scheme, which we refer to as \gn (GN), was discussed that can lead to further exponential decay of the error, provided a sufficiently good SNR is available for the forward channel. 
%
%
The SK scheme has been enhanced in various ways. The Modulo-SK scheme~\cite{ben2017interactive} and schemes by Chance-Love~\cite{chance2011concatenated} and Mishra et al~\cite{mishra2023linear}, extends the SK scheme for noisy feedback, while compressed error correction (CEC)~\cite{ooi1998fast} and accumulative iterative code (AIC)~\cite{perotti2021accumulative} focus on reducing channel use through continuous error vector compression, using noiseless feedback.



While guaranteeing impressive error exponents, SK and other analytical coding schemes have not been adapted to practical communication systems, as the improvement in performance does not justify the cost incurred in terms of high-numerical precision, increase in the amount of feedback, and large number of rounds of communication. Because of this, analytical feedback coding schemes in practice are limited to automatic repeat request (ARQ) and hybrid-ARQ (HARQ) retransmission schemes, where the receiver provides simple one-bit feedback to indicate success (acknowledgment/ACK) or failure (negative acknowledgment/NACK). Extending ARQ to multi-bit feedback is an interesting research topic; for example, compressed error hybrid ARQ (CE-HARQ)~\cite{ankireddy2023compressed} and Griffin et al.~\cite{griffin2023code} propose using full feedback from the receiver to iteratively improve the error vector at the receiver.  

Recent advances in deep learning revived the interest in coding for channels with feedback by leveraging the expressive power~\cite{raghu2017expressive} of deep neural networks.
Several works proposed deep learning approaches to improve the performance of channel codes, ranging from augmenting analytical decoders with learnable parameters~\cite{nachmani2016learning, nachmani2018deep, ankireddy2023interpreting, hebbar2022tinyturbo} to creating novel neural network architectures based neural encoders and decoders~\cite{hebbar2023crisp, hebbar2024deeppolar} and improving the code design using sequential models~\cite{li2021learning,ankireddy2024nested}. 
For channels with feedback, deep learning techniques were used to design new encoding and decoding schemes that take advantage of the high-capacity feedback channel. Deepcode~\cite{kim2018deepcode} modeled both encoder and decoder functions as recurrent neural networks (RNNs) to process a bit stream in a sequential manner to directly minimize the end-to-end transmission error over an additive white Gaussian noise (AWGN) channel.
Several works further explored this idea of using learning-based approaches for modeling the encoding and decoding operations, which can be broadly classified into the RNN family of codes~\cite{safavi2021deep,mashhadi2021drf,kim2023robust} and the transformer family of codes~\cite{shao2023attentioncode,ozfatura2022all}, discussed in detail in~\secref{sec:neural_codes}. The current state-of-the-art is generalized block attention feedback (GBAF)~\cite{ozfatura2022all}, which uses self-attention and transformer architecture to perform block coding.

While these state-of-the-art deep-learning-based feedback codes provide tremendous improvements in BLER performance, they also come with significant memory and computational costs that may not be supported by the next generation of communication transceivers that operate with limited onboard resources. Therefore, it is desirable to develop lightweight codes that use simple schemes while providing desirable performance. To accomplish a reduced complexity coding scheme, we explore two different directions in this work. First, we review and understand the existing analytical feedback coding schemes to identify the limitations and propose a new feedback coding scheme that provides non-trivial performance improvements over the existing schemes for channels with passive, noiseless feedback. Next, we propose a lightweight deep-learning-based feedback coding scheme that can significantly reduce the complexity compared to neural block-feedback coding schemes by limiting ourselves to a symbol-by-symbol scheme instead of block coding schemes. 



Our main contributions are summarized as follows:

\begin{itemize}
    \item We provide a comprehensive review of the existing analytical and deep-learning-based coding schemes for channels with feedback and identify the limitations of existing approaches (\secref{sec:analytical} and \secref{sec:neural_codes}).
    
    
    \item We propose \pb, an analytical coding scheme that noticeably improves the performance over Schalkwijk-Kailath and two-phase \gn schemes (\secref{sec:pb}) and exhibits reliability comparable to state-of-the-art deep-learning-based coding schemes in regions of high-SNR (\secref{sec:results}).
    

    
    \item We propose \lightcode, a lightweight neural coding scheme that achieves a performance superior to current state-of-the-art feedback coding schemes using $10 \times$ fewer parameters and computational complexity (\secref{sec:deepcode}, \secref{sec:results}), while maintaining interpretability. 

    \item We analyze the representations learned by the \lightcode encoder using linear regression and draw comparisons with \pb. We also demonstrate that the relation between encoded symbols and the feedback is highly non-linear in the final rounds, underscoring the need for a deep-learning-based coding scheme (\secref{sec:analysis}).


\end{itemize}

%% file: sysmodel.tex
\section{System Model}\label{sec:sysmodel}

We consider a feedback coding scheme with an AWGN forward channel and an AWGN feedback channel with noisy passive feedback. The goal is to transmit a message vector $\mathbf{u} \in \{0,1\}^K$ of length $K$ to the receiver using a total of $D$ independent channel uses \textit{i.e.,} the noise across the rounds is independent and identically distributed (i.i.d). This results in an overall coding rate of $R=\sfrac{K}{D}$. In this work, we consider symbol-by-symbol coding, where the block of $K$ bits will be mapped to one symbol. Hence, the terms block and symbol can be used interchangeably, and the number of rounds of communication for a rate $\sfrac{K}{D}$ code is $D$.

\begin{figure}[!htb]
    \centering
 	\includegraphics[width=\linewidth]{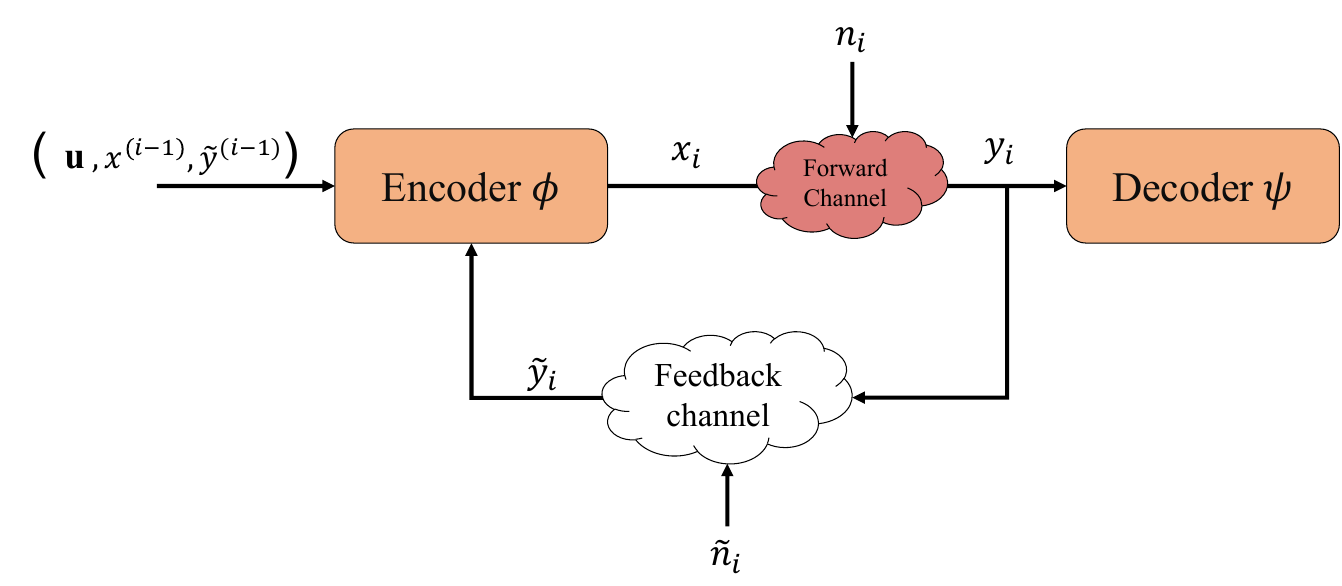}
 	\captionsetup{font=small}
 	\caption{Illustration of the $i^{\text{th}}$ round of communication for channels with feedback. The encoder takes as input the message bits $\mathbf{u}$ and the encoder output from previous rounds $x^{(i-1)}$, concatenated with the feedback from previous rounds $\Tilde{y}^{(i-1)}$, to compute $x_i$.}
    \label{fig:sysmodel}

\end{figure}

In the first round, the transmitter encodes the message block $\mathbf{u} \in \{0,1\}^K$ to a real-valued output $x_1 \in \mathbb{R}$ and transmits using the forward AWGN channel as
\begin{align}
    x_1 &= \phi(\mathbf{u}), \\
    y_1 &= x_1 + n_1,
\end{align}
where $\phi$ is the encoding function and $n_1 \sim \mathcal{N}(0,\sigma^2_{ff}) \in \mathbb{R}$ denotes the feedforward noise. The receiver then sends the noisy received symbol as feedback using the feedback channel as
\begin{equation}
    \Tilde{y}_1 = y_1 + \Tilde{n}_1,
\end{equation}
where $\Tilde{n}_1 \sim \mathcal{N}(0, \sigma_{fb}^2) \in \mathbb{R}$ denotes the feedback noise.

At round $i > 1$, the encoder $\phi$ computes the output $x_i$ using the input message bits $\mathbf{u}$ and the encoder outputs from previous $i-1$ rounds $x^{(i-1)} = \{ x_1, \dots, x_{i-1} \}$ and the feedback from previous $i-1$ rounds $\Tilde{y}^{(i-1)} = \{ \Tilde{y}_1, \dots, \Tilde{y}_{i-1} \}$ as 
\begin{equation}
    x_i = \phi \left(\mathbf{u}, x_1, \dots, x_{i-1}, \Tilde{y}_1, \dots, \Tilde{y}_{i-1} \right).
\end{equation}



At the end of $D$ rounds of communication, the decoder $\psi$ estimates the transmitted message vector $\hat{\mathbf{u}}$ using the received symbols from all $D$ rounds $\{ y_1, y_2, \dots, y_D \}$ as
\begin{equation}
    \hat{\mathbf{u}} = \psi \left( y_1, y_2, \dots, y_D \right),
\end{equation}
where $\psi$ is the decoding function and $\hat{\mathbf{u}} \in \{0,1\}^K$. The objective is to design an encoder-decoder pair $\{\phi,\psi\}$ that minimizes the probability of error $\Pr\{\mathbf{u} \neq \hat{\mathbf{u}}\}$ for a given number of rounds $D$, under a sum power constraint of $\sum_{i=1}^D \mathbb{E}[|x_i|^2]\ \leq\ D$.


%% file: analytical.tex
\section{Analytical Coding Schemes for Channels with Noiseless Feedback}\label{sec:analytical}
%
In this section, we review analytical coding schemes for channels with noiseless feedback \textit{i.e.,} $\sigma_{fb}^2 = 0$. We begin by reviewing the celebrated Schalkwijk-Kailath coding scheme~\cite{schalkwijk1966coding_p1}, one of the seminal works in coding for channels with noiseless feedback. We then review a less widely-known scheme by \gn~\cite{gallager2009variations}. This is similar to the SK scheme for the first $D-1$ rounds. Still, it deviates significantly in the last round, tailoring for the transmission of discrete messages, often significantly improving the performance in one round of communication. The \gn (GN) scheme is typically not considered as a baseline as it exhibits a worse error rate compared to SK in certain regions of SNR. However, in the next section, we introduce a novel analytical coding scheme, building on the SK and GN schemes. This new scheme achieves a significantly lower error rate than the SK or GN schemes and is often comparable to highly complex neural coding schemes in high-SNR regions, defying the conventional belief that neural coding schemes are much more reliable than analytical ones.

\subsection{Schalkwijk-Kailath coding scheme}
The SK scheme considers the problem of transmitting a fixed number of bits on an AWGN forward channel and a noiseless feedback channel. The transmission begins by mapping $K$ bits of information symbols to a single $2^{K}$-ary pulse-amplitude modulation (PAM) symbol $\Theta$ from the constellation
\begin{equation}
   \Theta \in \{ \pm 1 \eta, \pm 3 \eta, \dots, \pm (2^K -1) \eta\}, 
   \label{eqn:pam}
\end{equation}
where $\eta=\sqrt{\sfrac{3}{(2^{2K}-1})}$ is the scaling factor to ensure unit power normalization of the PAM constellation. Even though a block of information is transmitted since all the $K$ bits of information are mapped to one PAM symbol, the SK scheme at its core can be considered a symbol-by-symbol feedback scheme. We now describe the SK coding scheme in detail.

In the first round, the uncoded PAM symbol is transmitted after accounting for power constraint $P$ \textit{i.e.,} $x_1 = \sqrt{P} \Theta$ as
\begin{equation}
    y_1 = x_1 + n_1,
\end{equation}
where $n_1 \in \mathcal{N}(0,\sigma_{ff}^2)$ is the noise in the forward AWGN channel. The received symbol $y_1$ is then sent back to the transmitter noiselessly. For the second round, the transmitter first computes the receiver's estimate of the transmitted symbol based on $y_1$ as $\hat{\Theta}_1 = \frac{ \sqrt{P}y_1}{ P + \sigma^2_{ff}}$ using linear minimum mean-square error (LMMSE). 
In the second round, after receiving $y_1$ as feedback, the transmitter sends the scaled version of the {\em error} in the LMMSE estimate from the previous round, i.e., $\epsilon_1= \hat{\Theta}_1 - \Theta$. This process continues for the remaining rounds. 
In other words, starting from round $i=2$, the goal of the transmitter is to transmit the error in estimate from the previous round, $\epsilon_{i-1} = \hat{\Theta}_{i-1} - \Theta$, after scaling appropriately to satisfy the power constraint. The complete algorithm is described in~\algref{alg:SK}.
%

\medskip

\noindent\textbf{Error analysis.\ } 
It is shown in~\cite{ben2017interactive} that the probability of error for rate $\sfrac{K}{D}$ SK scheme is given by
\begin{equation}
    p_{SK} =  2(1 - 2^{-K})Q\left(   \sqrt{ \frac{3S(1+S)^{D-1}}{2^{2K}-1}     }       \right),
\end{equation}
where $S$ denotes the SNR of the forward AWGN channel on a linear scale. Note that~\cite{ben2017interactive} assumes a minimum-variance unbiased estimator (MVUE) at the end of round $1$, $\hat{\Theta}_1 = \frac{y_1}{\sqrt{P}}$, for ease of analysis but it is sub-optimal in terms of error in the estimate after round 1.








\RestyleAlgo{ruled}
\SetAlgoNlRelativeSize{-1}
\LinesNumbered
\SetKwComment{Comment}{/* }{ */}
\SetKwInput{KwInput}{Input}
\SetKwInput{KwOutput}{Output}
\SetKwInput{KwData}{Data}
\SetKwInput{KwIn}{In}
\SetKwInput{KwOut}{Out}

\setcounter{algocf}{0}

\begin{algorithm}[hbt!]
\caption{Schalkwijk-Kailath (SK) coding scheme}
\label{alg:SK}

\KwInput{Message symbol $\Theta$, number of rounds $D$, forward noise variance $\sigma_{ff}^2$ }

\vspace{5pt}
\textbf{Round 1:}\
\textbf{Tx:} 
Power normalization: $x_1 = \sqrt{P}\Theta$\;
Forward channel: $y_1 = x_1 + n_1$\;
\textbf{Rx:} LMMSE estimate of transmit symbol $\hat{\Theta}_1 = \frac{\sqrt{P}y_1}{P+\sigma^2_{ff}};$ 

\vspace{5pt}
\texttt{/* Tx communicates the error in estimate $\hat{\Theta}_1 - \Theta$ over the next $D-1$ rounds */} \\
\vspace{5pt}

\While{$2 \leq i \leq D$}{
    \textbf{Tx:} Compute the error in estimate of previous round  ${\epsilon}_{i-1} = \hat{\Theta}_{i-1} - \Theta$ \;
    Power normalization: $x_i = \frac{\sqrt{P}}{\sigma_{i-1}} {\epsilon}_{i-1}$, $\sigma^2_{i-1} = \mathbb{E}[\epsilon^2_{i-1}]$\;
    Forward channel:  $y_{i} = x_{i} + n_{i}$\;
    \textbf{Rx:} LMMSE estimate of transmit symbol $\hat{\epsilon}_{i-1} = \frac{\sqrt{P}\sigma_{i-1}}{P + \sigma_{ff}^2} y_{i}$ \;
    \vspace{5pt}
    Update the estimate of $\Theta$ as: $\hat{\Theta}_i = \hat{\Theta}_{i-1} - \hat{\epsilon}_{i-1}$
}
\vspace{5pt}
\textbf{Decoding}: Map $\hat{\Theta}_D$ to the closest symbol in the $2^K$ PAM constellation.

\end{algorithm}




\subsection{\gn coding scheme}

In~\cite{gallager2009variations}, \gn proposed a two-phase scheme that builds on the Elias scheme~\cite{Elias1956ChannelCapacity}, also closely related to the SK scheme. While the SK scheme considers the problem of transmitting a discrete symbol, Elias studied the problem of transmitting a Gaussian random variable $U \sim \mathcal{N}(0,\sigma^2)$. The strategy for forward and feedback transmissions is similar to SK, where a scaled version of the error in the LMMSE estimate is transmitted in every round $ I>1$. The main difference between the SK and Elias schemes lies in that the SK scheme aims to refine the message itself, while the Elias scheme aims to refine the estimate of noise added in the very first transmission.


GN scheme operates in two phases and 
relies on the assumption that after a sufficiently large number of rounds of the Elias scheme, 
the effective SNR for the forward channel shall be adequate for the noise variance to be considered small compared to the distance between the symbols in the PAM constellation. 
In such a high-SNR regime, a strategy superior to the Elias scheme can be implemented by taking advantage of the discrete nature of the signal. Instead of transmitting the original error vector 
with respect to $2^K$-ary PAM, 
the integer difference between the PAM index of the estimate and the true PAM symbol is transmitted. 
We refer to this as {\em discrete-symbol} scheme. 
This method results in an error exponent that decreases with an exponential \textit{order}, which increases linearly with the number of rounds. For this work, we assume that the high-SNR region is realized in the final round of communication, which is valid according to the 2-phase strategy described in~\cite{gallager2009variations}.

We now describe the GN scheme in detail. 
The first round of GN is simply uncoded PAM, except for the power allocation. In~\cite{gallager2009variations}, it is shown the optimal power distribution across $D$ rounds is attained by choosing $P_1$ and $P_2$ such that $P_1 + (D-1)P_2 = DP$ and $P_1 = P_2 + 1$, $P_1$ is the power constraint in round 1 and $P_2$ is the power constraint in remaining rounds. Hence, the transmission in round 1 is given by
\begin{equation}
    y_1 = \sqrt{P_1}\Theta + n_1.
\end{equation}
For the remaining rounds, the goal of the transmitter is to communicate the noise $n_1$ to the receiver as in the Elias scheme, 
where the LMMSE estimate at the receiver is improved iteratively, and a maximum likelihood (ML) detection is used at the end of $D-1$ rounds of communication to map the estimate to the original PAM constellation. 

Finally, for the last round, GN uses a discrete-symbol scheme suitable for the high-SNR region by transmitting the error in the PAM index $U$. We follow the assumption from~\cite{gallager2009variations} that the high-SNR region guarantees that $U \in \{-1, 0, 1\}$ with high probability.
In~\cite{gallager2009variations}, an ML decoder is used for the ease of analysis across multiple rounds of discrete symbol schemes, which is sub-optimal. However, since we assume only one round of the discrete-symbol scheme in this work, we assume a maximum-a-posteriori (MAP) decoder, which is optimal for the performance. Hence, the final round of GN can be viewed as MAP detection on a constellation $\{-1, 0, 1\}$ with probability distribution $\{\sfrac{p_{GN1}}{2}, 1 - p_{GN1}, \sfrac{p_{GN1}}{2}\}$,  $p_{GN1}$ is the probability of error after $D-1$ rounds. The complete algorithm is described in~\algref{alg:GN} in Appendix,~\secref{app:alg_gn}.



\medskip

\noindent\textbf{Error analysis.\ } The probability of error for a rate $\sfrac{K}{D}$ \gn scheme can be computed in two phases. The first phase can be analyzed as one round of uncoded PAM followed by $D-2$ rounds of Elias scheme, for which it is shown in~\cite{gallager2009variations} that the probability of error is given by
\begin{equation}
    p_{GN1} =  2(1 - 2^{-K})  Q\left(   \sqrt{ \frac{3(1+S - \frac{1}{D-1})^{D-1}}{2^{2K}-1}     }       \right),
\end{equation}
where $S$ denotes the SNR of the forward AWGN channel on a linear scale.

In the second phase, which corresponds to the final round, the discrete integer difference between the PAM index of the decoded message $\hat{M}$  and the index of the true message $M$ is transmitted. Here, the message index 
$M \in \{0, 1, 2, \dots, 2^K-1\}$ is deterministic based on the transmitted symbol $ \Theta \in \{ \pm 1 \eta, \pm 3 \eta, \dots, \pm (2^K -1) \eta\}$ and can be computed using a predetermined mapping, where $\eta$ is the scaling factor to ensure unit power normalization.  Similarly, $\hat{M}$ corresponds to the message index whose corresponding symbol is closest to the estimated transmitted symbol $\hat{\Theta}$.

As the error in the $(\hat{M} - M)$ lies in $ \{-1, 0, 1\}$ with probability distribution $\{ \sfrac{p_{GN1}}{2}, 1 -  p_{GN1},  \sfrac{p_{GN1}}{2}\}$ based on the assumption from~\cite{gallager2009variations}, the probability of error in decoding using a MAP decoder can be computed as
\begin{equation}
            p_{GN} = 2(1-p_{GN1})Q \left( \gamma \sqrt{S} \right) + p_{GN1}Q \left( \left( \frac{1}{p_{GN1}} - \gamma\right) \sqrt{S} \right),
\end{equation}
%
where $S$ is the SNR of the forward AWGN channel on a linear scale, and $\gamma$ is the detection threshold given by
\begin{equation}
            \gamma = \frac{1}{2 \sqrt{p_{GN1}}} + \frac{\sqrt{p_{GN1}}}{S} \log \left( \frac{2(1-p_{GN1})}{p_{GN1}}  \right). 
\end{equation}


%% file: powerblast.tex
\section{Proposed analytical coding scheme: \pb}\label{sec:pb}
In this section, we propose \pb, a hybrid 2-phase scheme, to iteratively refine the LMMSE estimate of the PAM symbol at the receiver and shift to a discrete symbol scheme that takes advantage of the sparsity of the error in the estimate of the PAM index in the final round. Through theoretical analysis and empirical results, we demonstrate that the performance of \pb is better than both SK and GN in several regimes of interest. 

For a rate $\sfrac{K}{D}$ code, \pb begins by mapping $K$ bits of information to a $2^K$ PAM constellation and transmits the symbols on the forward AWGN channel. The receiver performs an LMMSE estimate and sends the estimate as feedback to the transmitter through the noiseless feedback channel. This continues for the first $D-1$ rounds. In the final round, \pb uses a discrete symbol strategy to send the error in the PAM index of the estimate. The complete algorithm is described in~\algref{alg:PB}. 

\begin{algorithm}[hbt!]
\caption{\pb coding scheme}
\label{alg:PB}

\KwInput{Message symbol $\Theta$, number of rounds $D$, forward noise variance $\sigma_{ff}^2$ }

\vspace{5pt}
\textbf{Round 1:}\
\textbf{Tx:} 
Power normalization: $x_1 = \sqrt{P}\Theta$\;
Forward channel: $y_1 = x_1 + n_1$\;
\textbf{Rx:} LMMSE estimate of transmit symbol $\hat{\Theta}_1 = \frac{\sqrt{P}y_1}{P + \sigma_{ff}^2};$ 

\vspace{5pt}
\texttt{/* Tx communicates the error in estimate $\hat{\Theta}_1 - \Theta$ over the next $D-2$ rounds */} \\
\vspace{5pt}

\While{$2 \leq i \leq D-1$}{
    \textbf{Tx:} Compute the error in estimate of previous round  ${\epsilon}_{i-1} = \hat{\Theta}_{i-1} - \Theta$ \;
    Power normalization: $x_i = \frac{\sqrt{P}}{\sigma_{i-1}} {\epsilon}_{i-1}$, $\sigma^2_{i-1} = \mathbb{E}[\epsilon^2_{i-1}]$\;
    Forward channel:  $y_{i} = x_{i} + n_{i}$\;
    \textbf{Rx:} LMMSE estimate of transmit symbol $\hat{\epsilon}_i = \frac{\sqrt{P}\sigma_{i-1}}{P + \sigma_{ff}^2} y_{i}$ \;
    \vspace{5pt}
    Update the estimate of $\Theta$ as: $\hat{\Theta}_i = \hat{\Theta}_{i-1} - \hat{\epsilon}_{i-1}$
}

\vspace{5pt}
\textbf{Decoding after round $D-1$}: Map $\hat{\Theta}_{D-1} = \frac{\hat{x_1}}{\sqrt{P_1}\eta}$ to the closest symbol in the $2^K$ PAM constellation.

\vspace{5pt}
\texttt{/* High-SNR scheme for final round */} \\
\vspace{5pt}

\textbf{Round D:}\
\textbf{Tx:} Compute the difference in PAM indices $U = \hat{M} - M$, where $\hat{M}$ and $M$ correspond to the integer index from PAM constellation for $\hat{\Theta}$ and $\Theta$ respectively\;
Power normalization: $x_D =  \frac{\sqrt{P}U}{\sigma_{ff}}$;
Transmit: $y_D = x_D + n_D$\; 
\vspace{5pt}
\textbf{Final Decoding}: Use MAP decoder to detect $\hat{U}$ and detect original PAM signal using $\hat{U}$.

\end{algorithm}

We note that \pb can be seen as a combination of SK and GN. Rate $\sfrac{K}{D}$ \pb can be interpreted as $D-1$ rounds of rate $\sfrac{K}{D}$ SK scheme with LMMSE estimate, followed by the discrete variant of GN scheme for the final round, resulting in a non-negligible improvement in performance in certain regions of SNR. The error analysis for the \pb scheme 
is provided below. 








\begin{theorem}[Error Analysis]
The probability of error for rate $\sfrac{K}{D}$ \pb scheme is given by
\begin{equation}
    p_{PB} = 2(1-p_{SK})Q \left( \gamma \sqrt{S} \right) + p_{SK}Q \left( \left( \frac{1}{p_{SK}} - \gamma\right) \sqrt{S} \right)
\end{equation}
where \begin{equation}
    p_{SK} =  2(1 - 2^{-K})Q\left(   \sqrt{ \frac{3S(1+S)^{D-2}}{2^{2K}-1}     }       \right)
\end{equation}
and $\gamma$ denotes the detection threshold and $S$ denotes the SNR of the forward AWGN channel on a linear scale, and $Q()$ is the standard Q-function.
\end{theorem}

\begin{proof}
        We begin by computing the probability of error for phase 1, which consists of $D-1$ rounds of the SK scheme. For ease of analysis, we follow the same assumption of ~\cite{ben2017interactive}, where the estimator after round one is assumed to be MVUE instead of LMMSE, implying,
        \begin{equation}
            \hat{\Theta}_1 = \frac{y_1}{\sqrt{P}}.
        \label{eq:mvue}
        \end{equation}
        Further, based on~\algref{alg:SK}
        , we can view the effective channel corresponding to rounds 1 to $D-1$ as one round with effective SNR $S.(1+S)^{D-2}$, $S = \sfrac{P}{\sigma_{ff}^2}$ is SNR for the forward channel in linear scale. It is straightforward to show that the probability of error incurred in transmitting a symbol from unit power normalized $2^K$ PAM, using an optimal detector, is
        \begin{align*}
            p_e &= 2(1 - 2^{-K})Q\left(\frac{\sqrt{P}\eta}{\sigma} \right)\\
                &=  2(1 - 2^{-K})Q\left(    \sqrt{ \frac{3S}{2^{2K}-1}     }     \right),
        \end{align*}
        where $S$ is the forward channel SNR in linear scale and $Q()$ is the Q-function. Hence, the effective probability of error after $D-1$ rounds of SK is
        \begin{equation}
            p_{SK} = 2(1 - 2^{-K})Q\left(   \sqrt{ \frac{3S(1+S)^{D-2}}{2^{2K}-1}     }       \right).
        \end{equation}
        
        We now proceed to compute the probability of error for phase two, which is transmitting the difference in integer (message) index over an AWGN channel. The key assumption for the high-SNR region is that the noise variance is small enough (compared to the effective forward SNR) that any errors beyond decoding to adjacent PAM symbols \textit{i.e.,} errors beyond 1 index difference are essentially negligible~\cite{gallager2009variations}. Hence, the problem can be rephrased as finding the probability of error to communicate a symbol from the constellation $\{ -1, 0, 1 \}$, with probability distribution $\{ \sfrac{p_{SK}}{2}, 1 - p_{SK}, \sfrac{p_{SK}}{2} \}$ using MAP decoder. This can be computed as
        \begin{equation}
            p_{PB} = 2(1-p_{SK})Q \left( \gamma \sqrt{S} \right) + p_{SK}Q \left( \left( \frac{1}{p_{SK}} - \gamma\right) \sqrt{S} \right),
        \end{equation}
where $\gamma$ is the detection threshold and $S$ is the SNR of the forward AWGN channel on a linear scale.
         
\end{proof}


\noindent \textbf{\pb vs. SK and GN  schemes.} 
The key difference between \pb and SK is in the last round of communication, where \pb shifts to a discrete symbol scheme, as seen in line 13 of~\algref{alg:PB}, by taking advantage of the high effective SNR over the initial $D-1$ rounds \textit{i.e.,} the (effective) noise variance is much smaller than the distance between adjacent PAM symbols. Hence, it is beneficial to now transmit the error in the PAM constellation index, which lies in $\{-1, 0, 1\}$ with a very high probability, resulting in a much lower probability of decoding error.

Further, compared to GN, the key difference in \pb is the choice of information transmitted after round 1. GN communicates and refines the noise from round 1, $n_1$, starting round 2, as seen in line 4 of~\algref{alg:GN}. In contrast, \pb transmits and refines the error in LMMSE estimate of $\Theta$, as seen in line 4 of~\algref{alg:PB}. It is seen from~\cite{ben2017interactive} that the effective SNR (in linear scale) after $D$ rounds is given by $S(1+S)^{D-1}$ for GN and $(1+S)^{D} - 1$ for \pb. 
While the difference in effective SNR becomes negligible for large $D$, it can be significant when operating for a finite number of rounds, resulting in the superior performance of \pb.


The performance comparison between the error expressions for a general SNR and rate is not straightforward (as they are represented in terms of Q functions). Instead, we limit our comparison to the canonical settings considered in the recent literature on channels with feedback (e.g.,~\cite{ozfatura2022all}), for rates $\sfrac{3}{6}$ and $\sfrac{3}{9}$ and plot the BLER performance in~\figref{fig:bler_gn_vs_pb}; we observe 
significant gains of \pb compared to both SK and GN in terms of the BLER performance.

\begin{figure}[ht]
\centering 
\begin{minipage}[b]{0.9\linewidth}
  \centering
  \subfloat[Rate $R = \sfrac{3}{6}$]{%
    \includegraphics[clip,width=\linewidth]{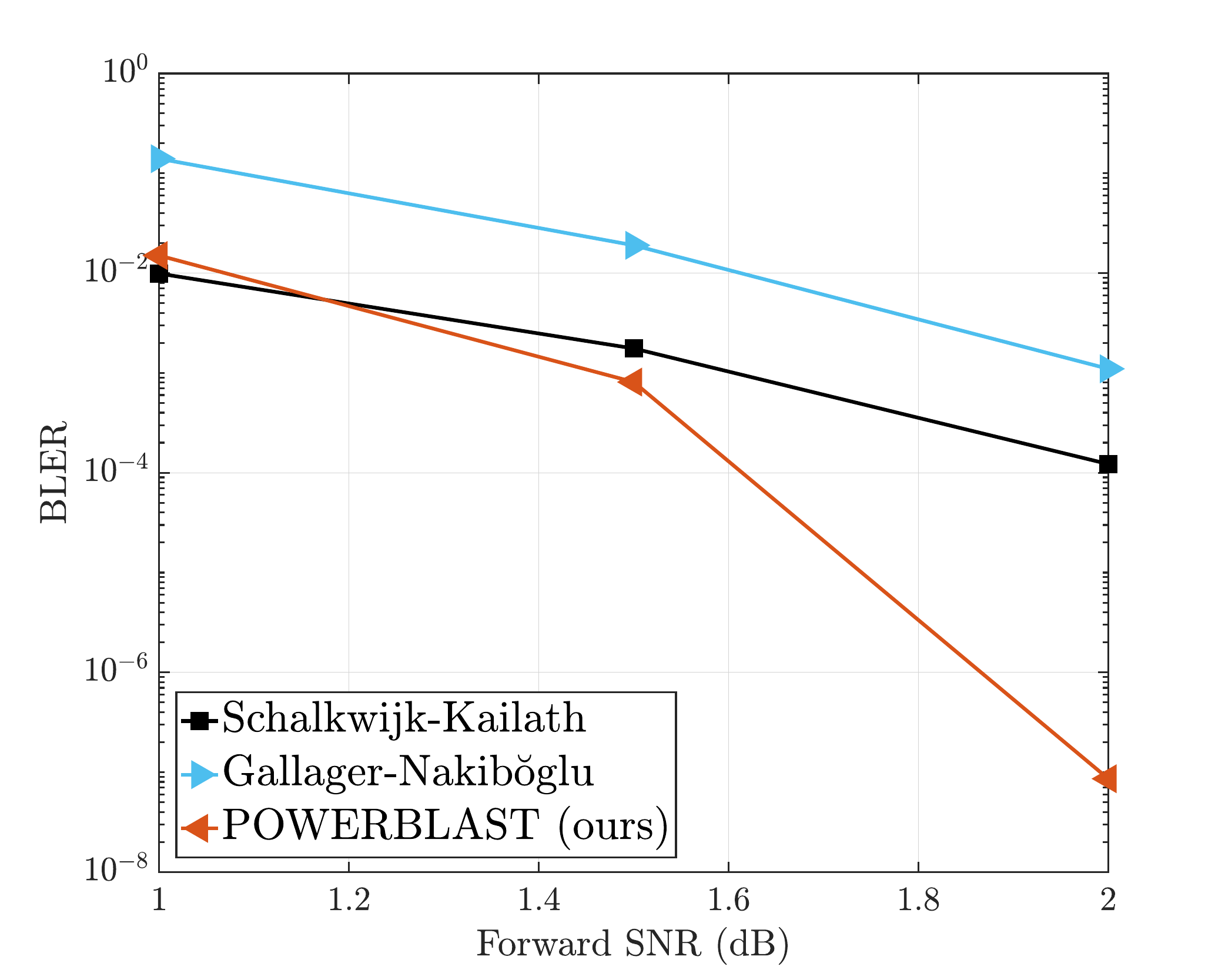}
  }
\end{minipage}
\vspace{0.5cm} 
\begin{minipage}[b]{0.9\linewidth}
  \centering
  \subfloat[Rate $R = \sfrac{3}{9}$]{%
    \includegraphics[clip,width=\linewidth]{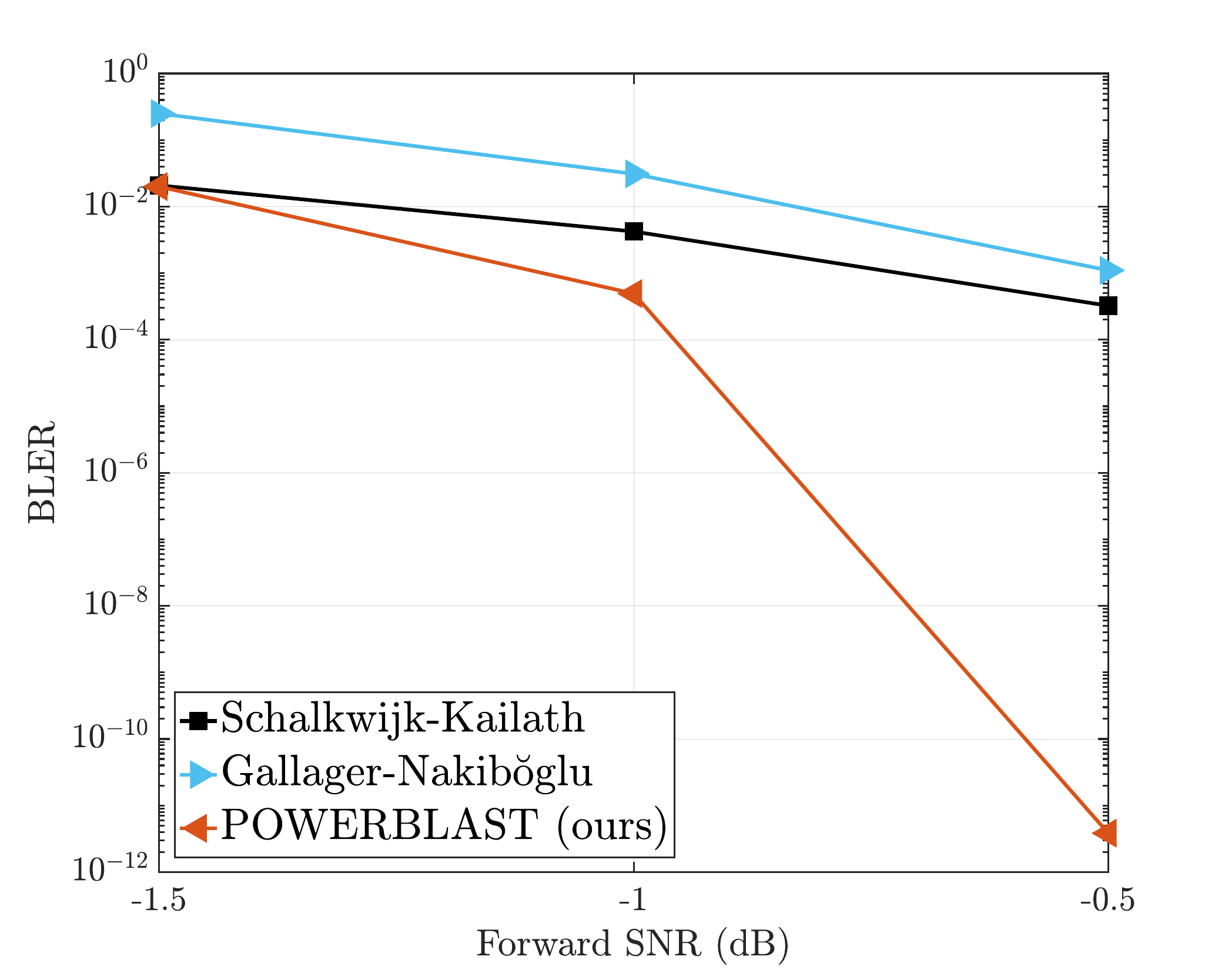}
  }
\end{minipage}
\captionsetup{font=small}
\caption{By combining the SK and discrete-symbol strategy of the GN scheme, \pb noticeably improves the BLER performance upon both SK and GN schemes.}
\label{fig:bler_gn_vs_pb}
\end{figure}




From the results so far, we have demonstrated that \pb provides the best performance among the analytically tractable solutions for channels with noiseless feedback that are recently considered in the literature (to demonstrate the reliability of deep-learning-based coding schemes~\cite{ozfatura2022all}). In the coming sections, we consider deep-learning-based coding schemes such as GBAF and show that, surprisingly, the analytical \pb coding scheme often delivers 
competing performance. Nevertheless, \pb still falls short when the SNR is very low; thus, we investigate conditions under which deep-learning-based schemes provide the highest gain. Finally, we propose \lightcode, a lightweight neural coding scheme that achieves state-of-the-art performance with very low complexity.

%% file: neural_fb_codes.tex
\section{Deep Learning based Coding schemes}\label{sec:neural_codes}


In this section, we review the current state of the learning-based codes for channels with feedback, closely analyze the generalized block attention feedback (GBAF)~\cite{ozfatura2022all}, discuss the shortcomings, and provide motivation for a new learning-based coding scheme we introduce in the next section.   


\subsection{RNN families}

Deep-learning-based algorithms are capable of modeling complex input-output relations. One of the critical challenges in designing codes for channels with feedback is accurately computing the subsequent transmission that minimizes the probability of decoding error, conditioned on the feedback from previous rounds. Existing classical coding schemes employ linear estimators at the receivers and model the dependencies in a linear fashion at the transmitter, which is sub-optimal. Instead, the sequential nature of the feedback from previous can be better leveraged by using deep-learning architectures such as RNNs tailored for processing sequential data. Deepcode~\cite{kim2018deepcode} demonstrated this advantage in modeling the dependencies across bits and rounds using bi-directional gated recurrent units (GRUs). Several follow-up works investigated the use of other RNN-based architectures such as long short-term memory  (LSTMs)~\cite{safavi2021deep,mashhadi2021drf} and, more recently, Robust Coding~\cite{kim2023robust} was proposed, which combined attention mechanism with bi-directional GRUs to optimize the symbol-by-symbol code design for noisy feedback channels across the rounds. Despite the promising performance, one disadvantage of using RNN-based architecture is the necessity to store the hidden states of the model, which are typically of much higher dimensions than the inputs and demand a lot of memory. Additionally, the sequential and iterative nature of encoding and decoding in RNNs can result in significant latency in the system.


\begin{figure*}[!htb]
    \centering
 	\includegraphics[width=0.9\linewidth]{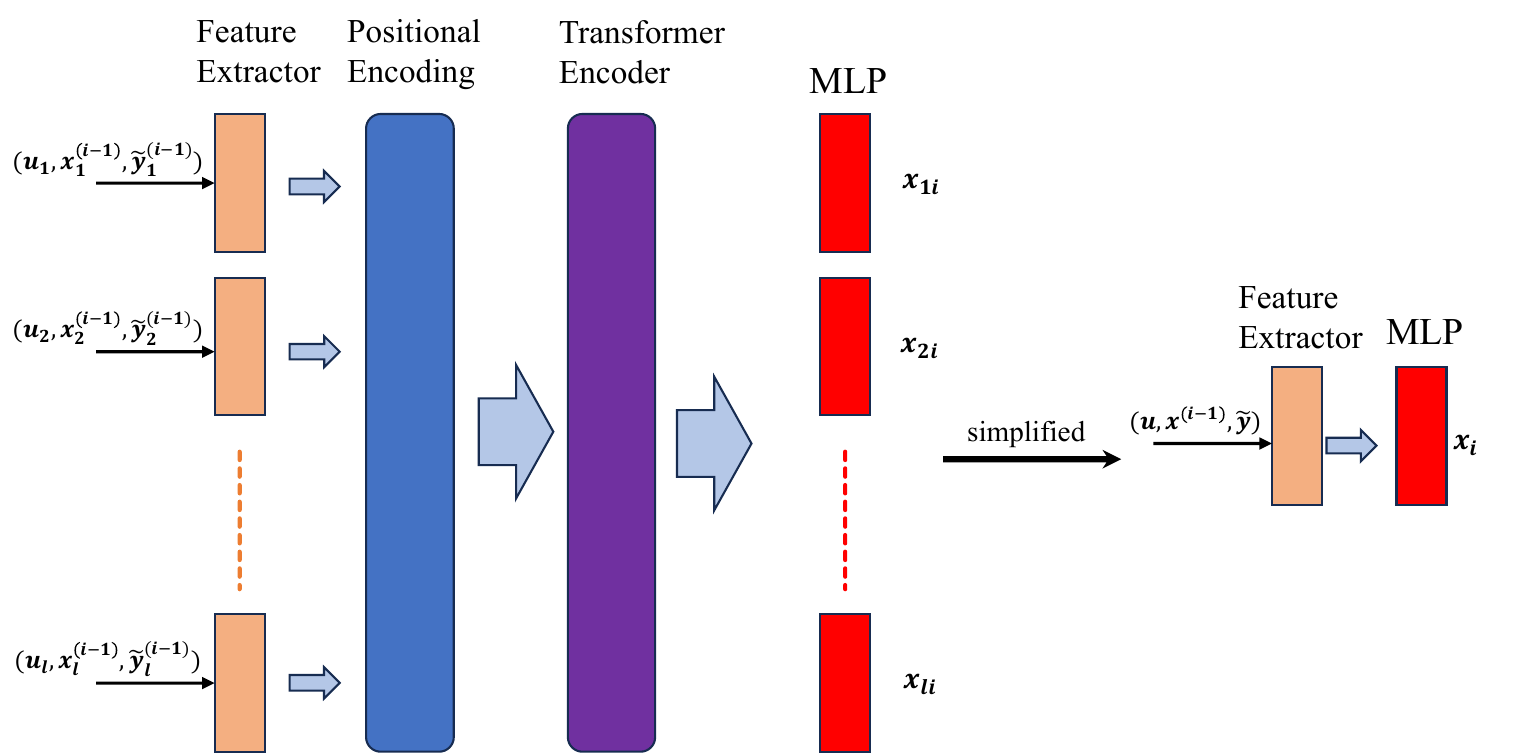}
 	\captionsetup{font=small}
 	\caption{
    (Left)Architecture for GBAF: The positional encoding and transformer encoding modules are used for block coding to encourage the mixing of symbols across the positions. (Right): Using a symbol-by-symbol scheme, \lightcode significantly reduces the complexity of encoding and achieves more than $10$x reduction in the number of parameters. On the left, we see the architecture for GBAF~\cite{ozfatura2022all}, and on the right, we see the architecture for \lightcode (ours). }
    \label{fig:simple_arch}
\vspace{-1em}
\end{figure*}

\subsection{Transformer families}\label{sec:gbaf}
In another line of work, self-attention-based transformer architectures have been explored for designing neural feedback codes. AttentionCode~\cite{shao2023attentioncode} introduced the idea of replacing RNN-based architecture with pure attention-based models, resulting in better alignment between a symbol and the corresponding feedback. In other words, AttentionCode can be viewed as Deepcode with transformer architecture. By leveraging the attention mechanism, AttentionCode creates temporal correlations at the transmitter
for encoding and exploits these temporal correlations
at the receiver for decoding. Further, the inputs to the encoder and decoder transformers are restructured to align each bit with the corresponding noise from multiple rounds as a single column, processed by the self-attention mechanism. More recently, GBAF~\cite{ozfatura2022all} introduced the idea of performing block coding across the codeblock, in addition to temporal coding across the rounds, resulting in orders of magnitude improvement in performance at extremely low SNRs, explained in detail below. 

As illustrated in \figref{fig:simple_arch}, the encoding of GBAF is performed across the rounds by causally concatenating the message and feedback symbols from previous rounds and using a series of feature extractors and multi-layer perception (MLP) modules. Additionally, positional encoding (PE) and a self-attention-based transformer encoder layer are deployed to encourage the mixing of information across the symbols within a codeblock, leading to block coding. 
More specifically, a block of $L$ bits is divided into $l$ sub-blocks of $K$ bits each and first encoded independently using a feature extractor. Next, PE and transformer encoder modules perform cross-symbol coding across the $l$ symbols. Finally, an MLP module is used to encode each symbol. A similar architecture is used at the decoder, but the output dimensions are adjusted accordingly. 

For concreteness, in~\cite{ozfatura2022all}, GBAF considers a block size $K=3$ and codeblock length of $L=51$ and performs a block coding across $l = 17$ symbols. This is the present state-of-the-art in performance, achieving a BLER of $7 \times 10^{-10}$ at SNR $-1.0$ dB for rate $\sfrac{3}{9}$.

\subsection{Important open problems} 


While providing impressive performance and reliability, transformer architecture is computationally expensive. Further, it is well known that the transformer architecture does not scale well to larger sequences. The self-attention mechanism imposes a compute complexity of $O(n^2)$ during training and $O(n)$ during inference, even after using the KV cache, with respect to input length $n$. Moreover, as the blocklength scales, a memory complexity of $O(n^2)$ prohibits training the algorithm from using a large batch size, which is crucial for attaining a good performance for any deep-learning models. 

An alternative approach would be to design a symbol-by-symbol coding scheme that can be scaled to any blocklength $L$ by encoding $K$ bits at a time independently. Within this context, an important question is: what is the extent of performance degradation compared to neural block coding schemes like GBAF? Moreover, it is necessary to formulate a novel and simplified architecture tailored to symbol-by-symbol processing. Finally, through streamlining the architecture and confining to symbol-by-symbol schemes, can we analyze and interpret the codes learned by deep learning models, discerning the reasons behind their notably superior performance compared to analytical counterparts?

In the coming sections, we answer all these questions. In \secref{sec:deepcode}, we systematically present \lightcode, a lightweight symbol-by-symbol neural coding scheme, which achieves state-of-the-art BLER performance~(\secref{sec:results}). Further, we analyze GBAF to study the efficacy of block coding by performing a systematic ablation study and also analyze \lightcode to identify the crucial components necessary for achieving ultra-low BLER in~\secref{sec:analysis}.









%% file: algorithm.tex
\section{Proposed neural coding scheme: \lightcode}\label{sec:deepcode}


Our goal is to design lightweight neural codes without the need for block coding \textit{i.e.,} we limit ourselves to symbol-by-symbol schemes suitable for both noiseless and noisy feedback settings. To this end, we design a lightweight deep-learning-based scheme with $10 \times$ fewer parameters compared to existing schemes. Surprisingly, this low-complex solution achieves a performance superior to current state-of-the-art deep-learning-based block-coding schemes. We now present the architecture and training choices crucial to achieving this new state-of-the-art BLER performance.

\subsection{\lightcode: Architecture}\label{sec:arch} 




\textbf{Our architecture.} We split the design of the encoder-decoder architecture into two parts: the feature extractor and the multi-layer perception (MLP) module, as illustrated in~\figref{fig:simple_arch} (right). The choice of the feature extractor plays a crucial role in determining the downstream performance. A higher complexity feature extractor can perform better but might not be desirable for practical applications. For \lightcode, after experimenting with various choices, we found that the design illustrated in~\figref{fig:feature_extractor} gave the optimal trade-off in complexity vs. BLER performance.
Further, the output of the feature extractor will be passed to an MLP module.
For the encoder, we choose a $1$ layer MLP to project the features to a $1$ dimensional output. For the decoder, we choose a $2$ layer MLP to transform the features into an output of dimension $2^K$. The full architecture for encoder and decoder, including the MLP are provided in~\appref{app:arch},~\figref{fig:encoder_decoder_arch}.

\medskip

\textbf{Our architecture vs. RNNs.} A popular choice for modeling the cross-round relation in feedback coding has been the RNN family of architectures. The sequential nature of the data makes it suitable for RNNs, GRUs, LSTMs, and other similar architectures. While these architectures provided impressive performance, a significant drawback is the necessity to store the hidden states from previous encoding steps, which is a high-dimensional latent that requires a lot of memory. Hence, a feed-forward architecture, which does not need a hidden state, is a better choice for resource-constrained scenarios. 

\medskip

\textbf{Our architecture vs. transformers.} While transformer models overcome the issue of storing hidden states by using positional encoding and self-attention, they also come with great computational complexity. Moreover, recent transformer-based feedback schemes proposed block coding using self-attention, which requires $O(n^2)$ memory and compute with respect to codeblock length. But as evident from results~\secref{sec:results}, block coding does not seem to provide noticeable gains in the setting under consideration. Hence, we propose a symbol-by-symbol coding scheme that uses a significantly simpler feed-forward architecture.



{\em Complexity.\ }By eliminating the need for block coding and using a short code length, \lightcode avoids the high-complexity transformer module used in GBAF architecture and instead uses a simple feed-forward network.
By carefully designing the architecture suitable for a symbol-by-symbol scheme, we achieve a lightweight design with more than $10\times$ reduction in the number of parameters compared to the RNN family of schemes such as Robust Coding and transformer family of block coding schemes such as GBAF. Further, we show in~\secref{sec:comp_throughput} that this reduction in complexity provides up to $171 \times$ higher decoding throughput compared to Robust Coding and up to $10\times$ higher throughput compared to GBAF. 

{\em Feature extractor.\ }{\figref{fig:feature_extractor} depicts the architecture of the feature extractor, which is the backbone for both encoder and decoder models. Compared to the feature extractor used in GBAF, we introduce two changes to improve the performance and reduce the complexity. The first is to add a skip connection, which preserves the prior from input to the feature extractor better, similar to the design of \deeppolar~\cite{hebbar2024deeppolar}. Further, the hidden dimension is reduced from $64$ to $32$, decreasing the number of parameters.}


Specifically, the input to the feature extractor is processed through a sequence of three linear layers, each with a hidden dimension of $32$, and rectified linear unit (ReLU) activation functions are applied between the first two layers. In parallel, we add a skip connection from output of first layer to input of final layer, which 
reverses the sign of the representation to introduce variance in the information.  
%
%
%
%
%
%
%
%
%
These outputs are concatenated and passed through the final layer, which generates a 16-dimensional feature representation. 




{\em Training.\ }Short code length and lightweight architecture allow for training of \lightcode with an extremely large batch size of $10^5$, which is more than $10 \times$ larger than the maximum batch size suitable for GBAF codes. We refer to~\secref{sec:training} for a detailed discussion of the training details and~\secref{sec:ablation} for an ablation study on the role of batch size.

We now describe the detailed encoding and decoding procedure below.


\medskip

\textbf{Encoding.} At round $i$, the encoder takes as input the original message and any available feedback from previous $i-1$ rounds to compute $x_i$. Further, after every round, a power reallocation is done across the rounds, similar to Deepcode~\cite{kim2018deepcode}. This is also based on the theoretical justification that allocating more power to the initial rounds results in an optimal performance~\cite{gallager2009variations}. The resulting encoding process at round $i$ is
\begin{equation}
    x_i = \alpha_i \phi(\mathbf{u}, \Tilde{y}_1, \Tilde{y}_2, \dots, \Tilde{y}_{i-1}, 0, \dots, 0),
\end{equation}
where $ \Tilde{y}_{i} = x_i + n_i + \Tilde{n}_i$ is the feedback from the previous round and $\alpha_i$ is the scaling factor to ensure sum power constraint across the rounds $\sum_{i=1}^D \alpha_i^2 = D$. Here, we note that in order to keep the size of the input to the encoder constant, the input is padded with zeros where necessary.


\medskip

\textbf{Decoding.} At the end of $D$ rounds of transmission, the decoder uses all the received symbols and estimates the original message $\hat{\mathbf{u}}$ as
\begin{equation}
    \hat{\mathbf{u}} = \psi(y_1, y_2, \dots, y_D),
\end{equation}
where $y_i = x_i + n_i$ is the noisy received symbol in round $i$.

Here, $\phi$ and $\psi$ are the encoder and decoder neural networks based on the architecture in~\figref{fig:simple_arch}. Using this simple symbol-by-symbol scheme, \lightcode archives a performance similar to that of GBAF. It turns out that the computationally intensive self-attention mechanism and the cross-symbol coding are not adding considerable value to GBAF, as will be evidenced by our ablation studies in~\secref{sec:ablation}.

\medskip

\textbf{Choice of feature vectors.} Empirically, we observed that the choice of inputs to the feature extractor has a noticeable effect on the BLER performance based on the feedback channel. For noiseless feedback, the input features $(\mathbf{u}, y_1, y_2, \dots, y_i)$ worked the best. Further, as will be evident from discussions in~\secref{sec:noisy}, the input features  $(\mathbf{u}, x_1, x_2, \dots, x_i, n_1+\Tilde{n}_1, n_2+\Tilde{n}_2, \dots n_i+\Tilde{n}_i)$ gave the best performance for noisy feedback. Apart from this difference in features, the rest of the architecture remains the same for noiseless and noisy feedback scenarios. 
\begin{figure}[tb]
    \centering
 	\includegraphics[width=\linewidth]{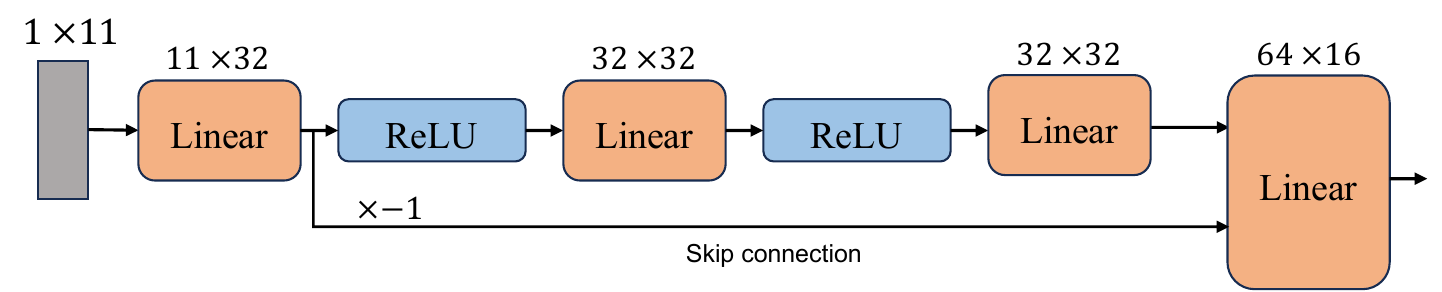}
 	\captionsetup{font=small}
 	\caption{Feature extractor design for \lightcode for a rate $\sfrac{3}{9}$ code.}
    \label{fig:feature_extractor}
\end{figure}

\medskip

\textbf{Support for multiple rates.} One of the limitations of existing symbol-by-symbol coding schemes is the inability to serve a variety of rates \textit{i.e.,} only rates of the form $\sfrac{1}{D}, D\in \mathbb{Z^+}$ are supported. In order to overcome this, instead of processing 1 bit at a time, \lightcode encodes a block of $K$ bits into $1$ symbol and communicates this symbol to the receiver over $D$ rounds. By independently varying $K$ and $D$ any rate of the form $\sfrac{K}{D}$ where, $K,D \in \mathbb{Z^+}$ can be supported, making the scheme more flexible, while simultaneously reducing the overall latency of the communication by a factor of $K$. The results for multiple rates for a block length of $K=3$ are discussed in~\tabref{tab:multi_rate}.

\subsection{Training}\label{sec:training}
Our primary region of interest is rate $\sfrac{3}{9}$ code at SNR $-1.0$ dB, where our target BLER is $\approx10^{-9}$. This is the current state-of-the-art performance by GBAF~\cite{ozfatura2022all}. To achieve such an extremely low error rate, it is important to train and simulate large amounts of samples reliably. Inspired by the high-SNR variant of \gn \cite{gallager2009variations}, we hypothesize that in regions of high-SNR or low errors, a significant benefit to the forward transmission arises from the power reallocation to symbols with non-zero errors. This can only be realized during the training by considering a large batch and enforcing the power constraint per batch.
Moreover, it is well understood that deep-learning models generalize better when trained with large batch sizes. Hence, we consider an extremely large batch size of $10^5$. While this is significantly larger than GBAF, which has a batch size of 8192, it is still a smaller number of symbols, considering that GBAF contains 17 symbols per codeblock. In contrast, our scheme consists of only $1$ symbol per codeblock. Further, it is possible to train with such a large batch size because of the small number of parameters compared to other deep-learning-based coding schemes. For a fair comparison, we follow a training methodology similar to that of GBAF, as explained below.  

\RestyleAlgo{ruled}

\SetKwComment{Comment}{/* }{ */}
\SetKwInput{KwInput}{Input}
\SetKwInput{KwOutput}{Output}
\SetKwInput{KwData}{Data}
\SetKwInput{KwIn}{In}
\SetKwInput{KwOut}{Out}

\begin{algorithm}[hbt!]
\caption{Training \lightcode}
\label{alg:lightcode}

\KwInput{Encoder model $\phi$, Decoder model $\psi$, Block length $K$, number of rounds $D$,  forward noise variance $\sigma_{ff}^2$, feedback noise variance $\sigma_{fb}^2$, batch size $B$, number of epochs $E$, learning rate $\texttt{lr}$ }
\vspace{5pt}
\For{$i \leq E$}{
    \vspace{5pt}
    Generate batch of random binary vectors $\mathbf{u} \in \{0,1\}^{K \times B}$\\
    \vspace{5pt}
    \For{$i \leq D$}{\Comment*[r]{Encoding at round $i$}
        \vspace{5pt}
        
        \If {$\sigma_{fb}^2 == 0$}{
            $x_i = \phi \left(\mathbf{u}, y_1, y_2, \dots, y_{i-1}, 0, \dots, 0 \right)$ }
        \vspace{5pt}
    
        \Else{
            $x_i = \phi \left(\mathbf{u}, x_1, \dots, x_i, n_1+\Tilde{n}_1, \dots n_i+\Tilde{n}_{i-1}, \dots, 0 \right)$ }
        $y_i = x_i + n_i$\vspace{5pt}} 
        
   $\mathbf{p_u} = \psi(y_1, y_2, \dots, y_D)$ \Comment*[r]{Decoding after $D$ rounds}
    \vspace{5pt}
    Compute the multi-class cross entropy loss $ \frac{1}{B} \sum_{j=1}^{B} \mathcal{L}_{\text{CE}}(\mathbf{c_j},\mathbf{p_{u_j}})$, $\mathbf{c_j}$ is the class index corresponding to $j^{\text{th}}$ message vector $\mathbf{u_j}$ and $\mathbf{p_{u_j}}$ is class probability vector after the \texttt{SoftMax} layer.\\
    \vspace{5pt}
    Clip the gradients to $0.5$.\\
    \vspace{5pt}
    Update model parameters for $\phi$ and $\psi$ using \texttt{AdamW} optimizer with learning rate \texttt{lr}.\\
    \vspace{5pt}
    Update the learning rate using \texttt{LambdaLR}. 
}

\end{algorithm}

We use \texttt{AdamW} optimizer, a stochastic optimization method that modifies the typical implementation of weight decay in \texttt{Adam} by decoupling weight decay from the gradient update. We initialize the learning to $10^{-3}$ and use a \texttt{LambdaLR} scheduler with a weight decay of $0.01$. Additionally, we clip all the gradient values to $0.5$ for numerical stability. Starting with randomly initialized weights, we jointly train the encoder and decoder models on $1.2\times10^5$ batches. Each input symbol corresponds to $K$ bits, resulting in a $2^K$ category classification problem for the decoder. Accordingly, we use the multi-class cross entropy (CE) loss to measure the performance as
\begin{equation*}
    \mathcal{L}_{\text{CE}} = \frac{1}{B} \sum_{i=1}^{B} \left(  \sum_{j=0}^{2^K - 1}  c_{ij} \log p_{ij} \right),
\end{equation*}
where $B$ is the batch size, $2^K$ is the number of classes , $c_{ij}$ is the true class probability and $p_{ij}$ is the predicted probability of the $j^{\text{th}}$ class, for the $i^{\text{th}}$ sample.


The hyperparameters used for training the rate $R = \sfrac{3}{9}$ code are listed in~\tabref{tab:hyperparameters}.

\begin{table}[!ht]
\centering

\begin{tabular}{lc}
\hline
\textbf{Hyperparameter} & \textbf{Value} \\
\hline
Encoder training SNR & -1.0 dB \\
Decoder training SNR & -1.0 dB \\
Mini batch size ($B$) & 100,000 \\
Total epochs ($E$) & 120 \\
Batches per epoch & 1000 \\
Optimizer & AdamW \\
Initial learning rate ($\texttt{lr}$) & $10^{-3}$ \\
Learning rate scheduler & LambdaLR \\
\hline
\end{tabular}
\captionsetup{font=small}
\caption{Hyperparameters for training rate $\sfrac{3}{9}$ \lightcode. }
\label{tab:hyperparameters}
\end{table}

Once the training is complete, we compute the mean power and standard deviation for the encoded data after every round for a large number of samples, $O(10^6)$, to reliably estimate the mean and standard deviation corresponding to encoder outputs to be used during inference. This is crucial in enforcing the power constraint in expectation. The algorithm for training \lightcode is described in detail in~\algref{alg:lightcode}.



%% file: results.tex

\section{Main Results }\label{sec:results}

We begin by comparing the performance of \lightcode with the current state-of-the-art in deep-learning-based coding schemes for noiseless passive feedback setting, GBAF~\cite{ozfatura2022all} and other schemes, demonstrated in~\figref{fig:bler_3_9_noiseless}. Next, we discuss a method for extending \lightcode to moderate block-length regimes of up to $L=51$ and study the performance.  

Finally, we look at deep-learning-based coding schemes for noisy feedback settings and how \lightcode can be extended to this setting with minimal changes. We then proceed to compare the BLER performance for the same configuration as before but with a feedback SNR of $20$ dB.


\subsection{\lightcode and \pb vs. existing neural codes for noiseless feedback}\label{sec:comparison}

In this section, we evaluate the performance of \lightcode and \pb and compare them against several existing analytical and deep-learning feedback schemes. For concreteness, we consider the canonical setting of rate $R=\sfrac{3}{9}$ with block length $K=3$ and $D=9$ rounds of communication on AWGN forward channel and noiseless feedback.


\medskip

\textbf{Baselines.} Our primary comparison is against GBAF~\cite{ozfatura2022all}, which is the current state-of-the-art for the noiseless passive feedback setting. Additionally, we consider Deepcode~\cite{kim2018deepcode}, DEFC~\cite{safavi2021deep}, DRFC~\cite{mashhadi2021drf}, Attentioncode~\cite{shao2023attentioncode}, and Robust Coding~\cite{kim2018deepcode}. Further, for completeness as well as to understand the relative gains of deep-learning-based coding schemes, we also compare the performance against NR-LDPC~\cite{Huawei2017LDPC}, Schalkwijk-Kailath~\cite{schalkwijk1966coding_p1}, \gn~\cite{gallager2009variations} and \pb. 

\textbf{Results.} In~\figref{fig:bler_3_9_noiseless}, we compare the BLER performance of rate $\sfrac{3}{9}$ coding schemes. We train a pair of encoder-decoder models at each SNR point in the plot.
\lightcode consistently outperforms the existing deep-learning-based schemes, including GBAF, while utilizing $< \sfrac{1}{10}^{\text{th}}$ the number of parameters. Interestingly, these results indicate that \pb surpasses the performance of all existing schemes, including \lightcode, when an adequately high signal-to-noise ratio is provided, which is $-0.5$ dB for rate $\sfrac{3}{9}$. However, the performance of deep-learning codes is significantly better at lower SNRs. Here, we highlight that by utilizing the clean feedback, \lightcode and \pb achieve an extremely small error rate even when operating at $0.75$ dB below the channel capacity. Further, we note that because of the extremely short block lengths considered, it is hard to derive meaningful achievability or converse bounds, such as~\cite{polyanskiy2010channel}, where blocks lengths of $100$ or higher are considered.

The blocklengths considered for the baselines in~\figref{fig:bler_3_9_noiseless} are listed in~\tabref{tab:block_lengths}, where $K=50$ for some of the schemes, which is larger than the length $K=3$ considered for rest of the schemes including \lightcode. While the BLER comparison against the different blocklengths is unfair, to maintain consistency, we used the same methodology followed in the current state-of-the-art results, including GBAF~\cite{ozfatura2022all} and Robust Coding~\cite{kim2023robust}. For a fairer comparison, we propose a modular approach for extending \lightcode to larger blocklengths in~\secref{sec:long_length}.

\begin{table}[!ht]
\centering
\begin{tabular}{lc}
\hline
\textbf{Scheme} & \textbf{Message length $K$} \\
\hline
Deepcode & $50$ \\
DEFC & $50$ \\
DRFC & $50$ \\
AttentionCode & $50$ \\
\gn & $3$ \\
\pb & $3$ \\
Robust Coding & $3$ \\
GBAF & $3$\footnotemark[1] \\
\lightcode & $3$ \\    
\hline
\end{tabular}
\captionsetup{font=small}
\caption{Block lengths of different coding schemes.}
\label{tab:block_lengths}
\end{table}


\footnotetext[1]2{While GBAF uses a blocklength of $51$, the BLER in~\cite{ozfatura2022all} was reported for the sub-block of length $3$.}


\begin{figure}[!htb]
    \centering
 	\includegraphics[width=\linewidth]{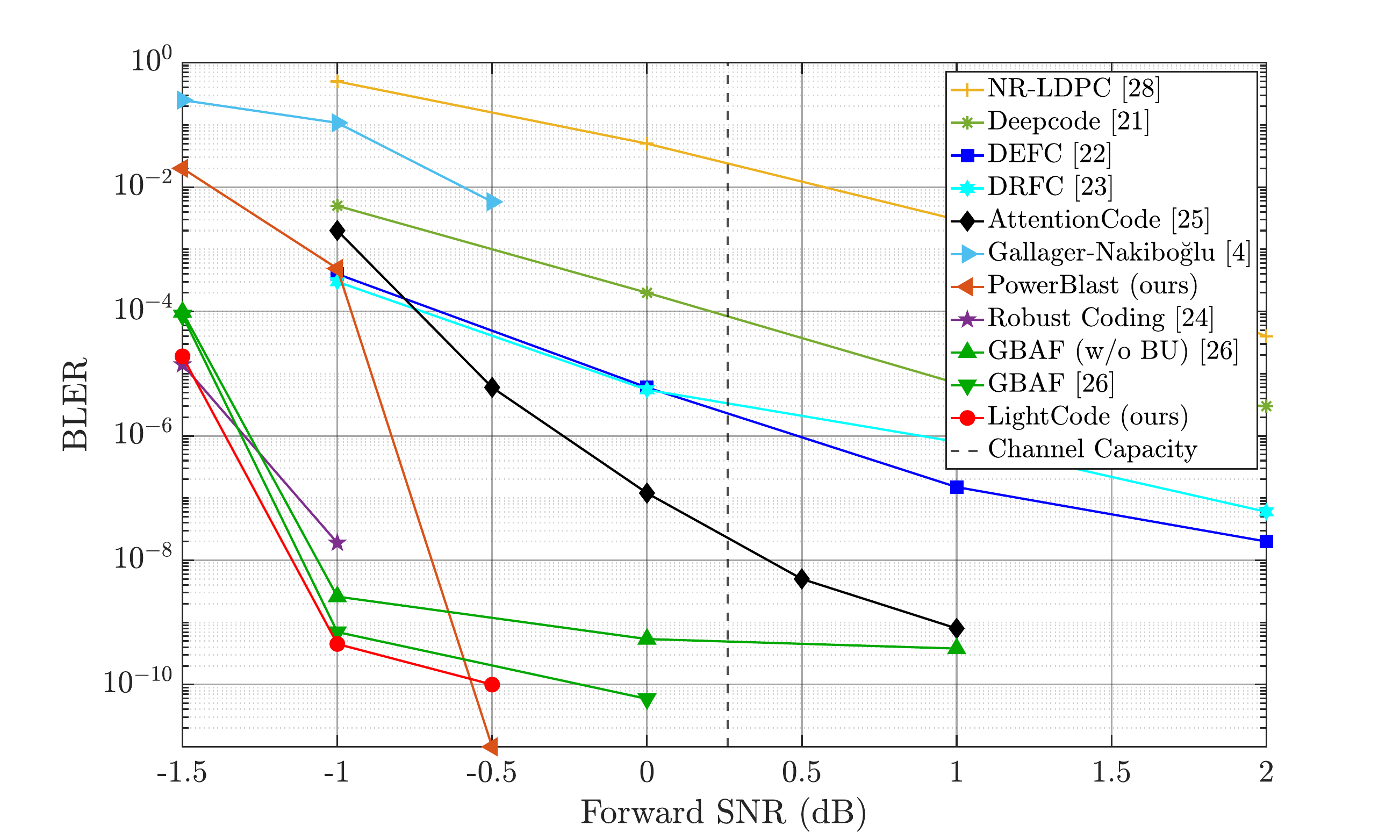}
 	\captionsetup{font=small}
 	\caption{Noiseless feedback: Performance comparison against existing classical and neural feedback codes for rate $\sfrac{3}{9}$. \pb achieves the best performance among existing classical schemes and performs comparable to state-of-the-art neural coding schemes in high-SNR regions. \lightcode achieves superior BLER performance compared to GBAF while utilizing $< \sfrac{1}{10}^{\text{th}}$ the number of parameters.}
 	\label{fig:bler_3_9_noiseless}
\end{figure}


\textbf{Multiple rates.} The coding rate of \lightcode is determined by two factors, block length $K$ and number of rounds of communication $D$. We keep the block length constant and vary $D$ to compare the performance against GBAF across multiple rates, as shown in~\tabref{tab:multi_rate}. \lightcode consistently outperforms GBAF across a range of rates $\{ \sfrac{3}{9}, \sfrac{3}{8}, \sfrac{3}{7}, \sfrac{3}{6}, \sfrac{3}{5}\}$, while utilizing a fraction of the number of parameters. 

\begin{table}[!htp]
\centering
\begin{tabular}{ccccc}
\toprule
\textbf{SNR (dB)} & \textbf{Rate} & \textbf{GBAF} & \textbf{\pb} & \textbf{\lightcode} \\
\midrule
$-1.0$ & $3/9$ & $7\times 10^{-10}$ & $2.8\times 10^{-4}$ & $\mathbf{4.5\times 10^{-10}}$ \\
$0.0$ & $3/8$ & $6.1\times 10^{-8}$ & $8.8\times 10^{-7}$ & $\mathbf{5.1\times 10^{-9}}$ \\
$1.0$ & $3/7$ & $7.5\times 10^{-8}$ & $1.0\times 10^{-8}$ & $\mathbf{1.0\times 10^{-8}}$ \\
$2.0$ & $3/6$ & $1.5\times 10^{-6}$ & $1.5\times 10^{-8}$ & $\mathbf{8.3\times 10^{-7}}$ \\
$3.0$ & $3/5$ & $8.7\times 10^{-7}$ & $1.9\times 10^{-4}$ & $\mathbf{2.7\times 10^{-7}}$ \\
\bottomrule
\end{tabular}
\captionsetup{font=small}
\caption{BLER performance comparison of \lightcode with GBAF with \pb for different rates. \lightcode consistently performs better than GBAF while using $< \sfrac{1}{10}^{\text{th}}$ the number of parameters.}
\label{tab:multi_rate}
\vspace{-1em}
\end{table}


\textbf{Error floor.} As observed in~\figref{fig:bler_3_9_noiseless}, \lightcode exhibits an error floor in the high-SNR region, similar to the existing deep-learning-based coding schemes such as Deepcode and GBAF. We believe a key reason for this behavior is the difficulty of training the decoder in the high SNR region, where the error events are extremely rare. For instance, while training in the regime of BLER $10^{-9}$, most of the training batches (batch size $10^5$) are error-free and do not provide sufficient guidance for the gradient descent to learn useful information, leading to saturation in performance.





\subsection{Performance in moderate blocklength regime}\label{sec:long_length}


 We now present a modular way to scale \lightcode to larger block lengths to enable a fair comparison against baselines with longer blocklengths. As a result of the curse of dimensionality, it is well known that the hardness of learning a code increases considerably with an increase in block length. As seen from results in Appendix~\secref{app:long_lightcode}, directly increasing the blocklength for \lightcode results in a poor BLER performance.
GBAF~\cite{ozfatura2022all} uses a computationally intensive transformer for performing block coding to support a length of $51$. To keep the complexity low, \lightcode instead uses a short blocklength $K$ (e.g., $K=3$) and treats a block of $L$ bits as $l = \sfrac{L}{K}$ independent sub-blocks. This modular approach of encoding $K$ bits at a time is motivated by ablation studies that demonstrate the negligible benefit of block coding in GBAF, presented in detail in~\secref{sec:ablation}. Thus, the resulting BLER for a blocklength $L$ \lightcode can be computed as
\begin{equation}
    p_L = 1 - (1-p_K)^l,
\label{eq:bler_modulo}
\end{equation}
where $p_L$ is the BLER for blocklength $L$ and $p_K$ is the BLER for blocklength $K$.

Finally, for a fair comparison against schemes with block length $K=50$, we refer to the results in~\figref{fig:bler_3_9_noiseless_51} where we compute the BLER for blocklength $L=51$ for \lightcode using the modular approach presented above. At a forward SNR of $-1.0$ dB, \lightcode with blocklength $51$ has a BLER of $7 \times 10^{-9}$ which is significantly smaller than Deepcode, DEFC, DRFC, and AttentionCode that use a block length of $50$, demonstrating the superior performance of \lightcode even at moderately longer blocklengths. We note that the BLER $p_L$ approaches $1$ for very large values of $L$, but we restrict our study to the relatively short block length regime of $L<300$, where the performance is still superior to the baselines as demonstrated in~\figref{fig:bler_3_9_noiseless_51}. Further potential performance improvements at longer blocklengths may be achievable through a concatenated coding scheme that employs an outer block code, as explored in~\cite{chance2011concatenated}; we leave this as a direction for future work. 

\begin{figure}[t]
    \centering
 	\includegraphics[width=\linewidth]{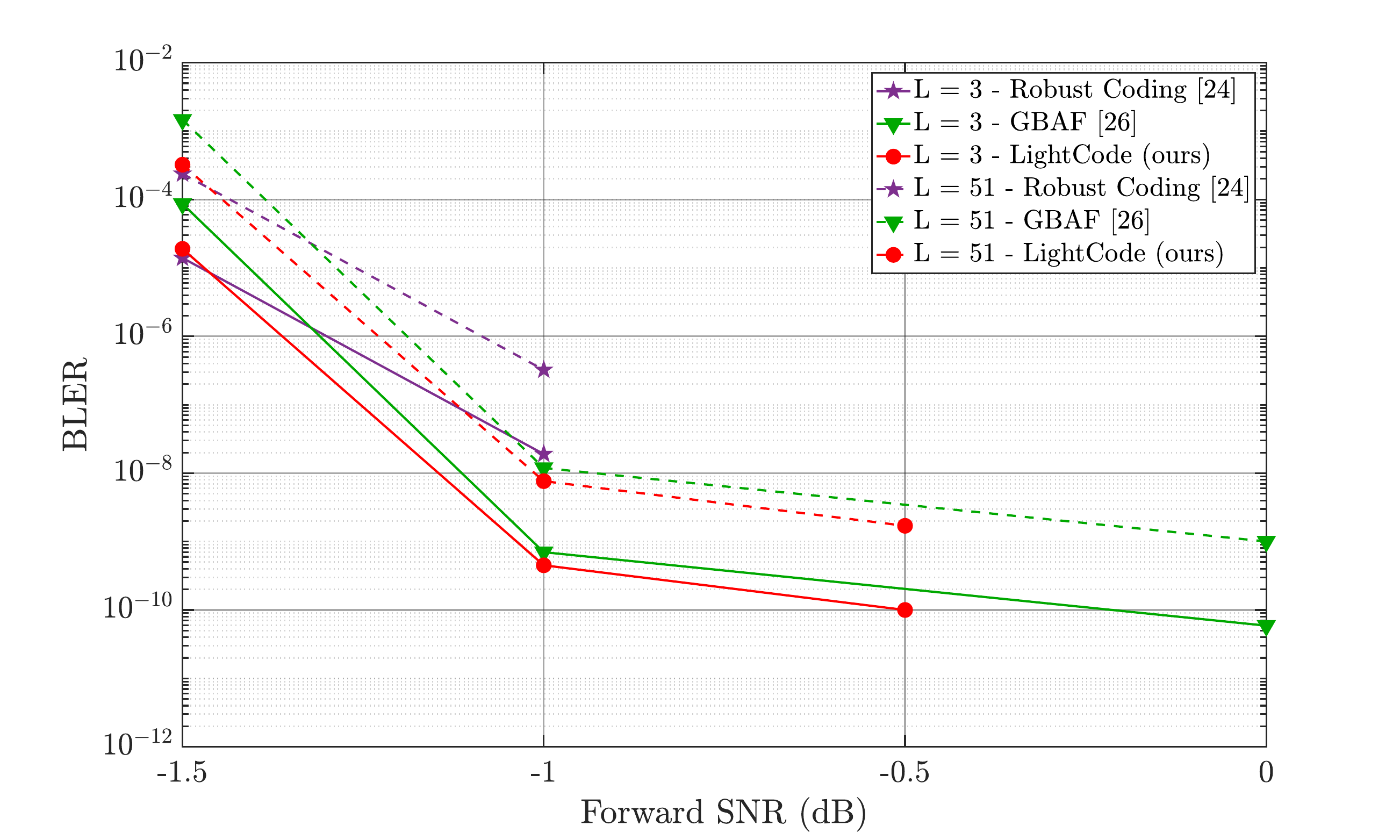}
 	\captionsetup{font=small}
 	\caption{By independently encoding sub-blocks of length $K=3$, \lightcode provides a flexible way to encode for any blocklength $L$, providing a steady trade-off in BLER performance vs blocklength. Even at a block length of $50$, \lightcode significantly outperforms all baseline methods shown in~\figref{fig:bler_3_9_noiseless},  which utilize a maximum block length of $51$.}
 	
\label{fig:bler_3_9_noiseless_51}
\end{figure}


\subsection{Coding for Channels with Noisy Feedback}\label{sec:noisy}
Finally, we now consider the case of channels with noisy feedback \textit{i.e.,} $\sigma_{fb}^2 > 0$. While the assumption of noiseless feedback is easier to study, most practical feedback channels suffer from noise even when the receiver sending the feedback operates at high power. Because of this limitation, neither the SK nor GN schemes perform well. To address this,~\cite{chance2011concatenated} introduces a linear feedback scheme that is implemented as an inner code to a concatenated code, which was found to be asymptotically optimal within the linear family of codes under AWGN forward channel~\cite{agrawal2011iteratively}. More recently,~\cite{mishra2023linear} proposed using dynamic programming to improve the performance, which turns out to be a generalized version of SK. However, despite these improvements, the linearity of these schemes severely limits the performance that can be achieved.

Recalling the architecture of deep-learning-based schemes introduced in~\secref{sec:neural_codes}, one of the interesting properties of these codes is their robustness to noise in the feedback channel. By taking advantage of the general architecture, injecting noise into the feedback channel during training is straightforward to make the encoder-decoder robust. Hence, learning-based schemes can help design a practically realizable class of codes that can be trained in a data-driven fashion and can provide gains in noiseless and noisy feedback settings.

\begin{figure}[t]
    \centering
 	\includegraphics[width=\linewidth]{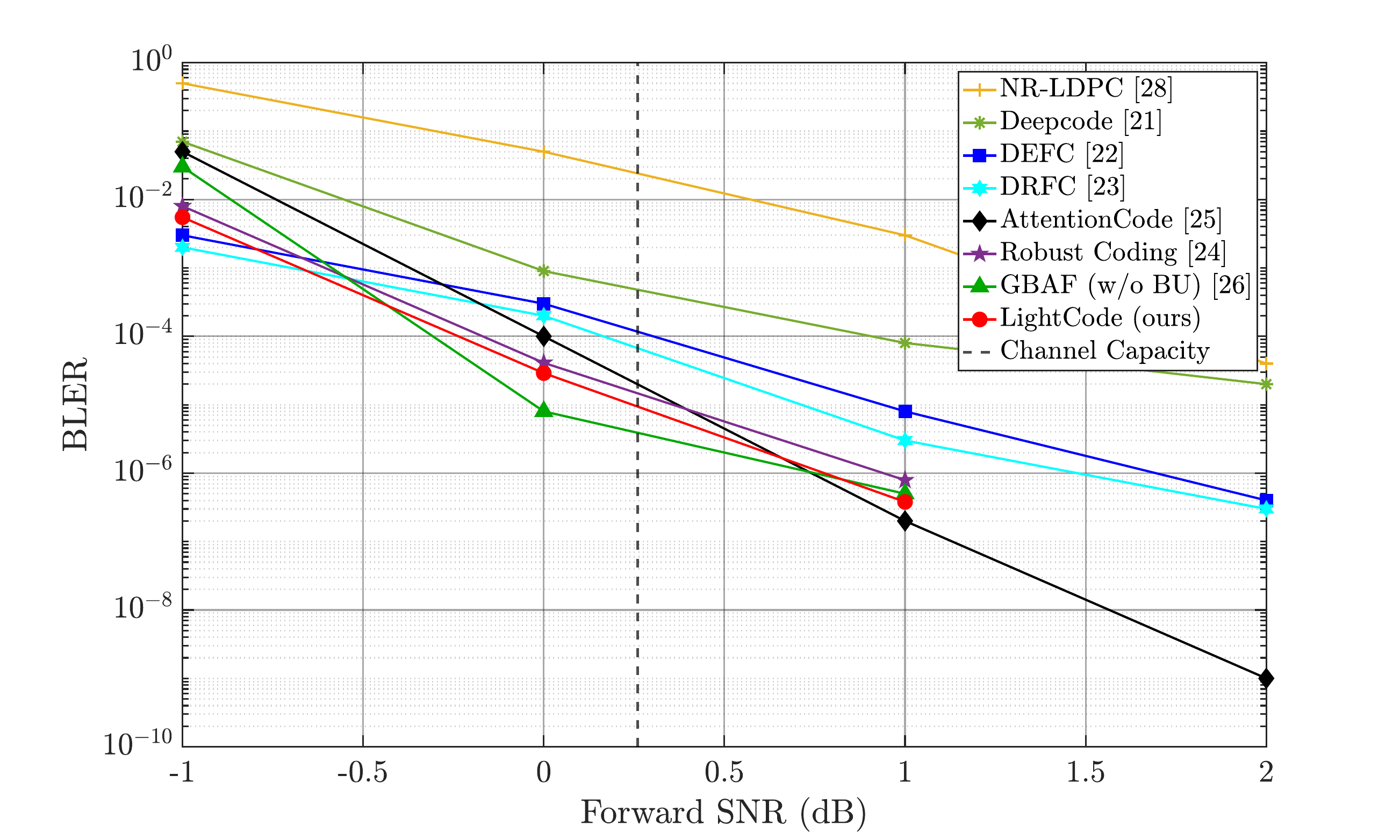}
 	\captionsetup{font=small}
 	\caption{\textcolor{black}{Noisy feedback: Performance comparison against existing neural feedback codes for rate $\sfrac{3}{9}$ and feedback SNR = $20$ dB. \lightcode achieves BLER performance comparable to GBAF while utilizing only $\sfrac{1}{10}^{\text{th}}$ the number of parameters.}}
 	\label{fig:bler_3_9_noisy}
\end{figure}

For training \lightcode on noisy feedback channels, empirically, we found that selecting the input features as the previous encoder outputs $x_1$ and the cumulative noise $n_i + \Tilde{n}_i$ separately works better than directly passing the feedback $\Tilde{y_i}$ from previous rounds, which was the choice for noiseless feedback. Except for the change in input feature, the rest of the architecture and training details remain the same as the noiseless feedback setting. In~\figref{fig:bler_3_9_noisy}, we compare the performance of rate $\sfrac{3}{9}$ \lightcode against existing schemes, for a feedback SNR of $20$ dB. We see that \lightcode exhibits similar performance compared to GBAF while still utilizing only a fraction of the number of parameters.

%% file: analysis.tex
\section{Analysis}\label{sec:analysis}
%

As evident from results in~\secref{sec:results}, symbol-by-symbol neural coding schemes such as \lightcode can achieve performance comparable to block coding schemes such as GBAF. Consequently, we examine GBAF and \lightcode in greater detail to systematically analyze their architecture, training process, and performance, with the aim of identifying the key factors contributing to their performance. Furthermore, we compare the memory and computational complexity of \lightcode with existing schemes to quantify the gains. Finally, we provide an interpretation of the encoded representations of \lightcode.

\subsection{Ablation studies on GBAF}\label{sec:ablation}

%

We investigate the contribution of the self-attention mechanism to the performance of GBAF by performing an ablation study to understand the compute vs performance trade-off better. First, the self-attention mechanism is disabled, and it is observed that this has no significant effect on the BLER performance. Next, both self-attention and positional encoding blocks are disabled, and the performance remains approximately the same. These results, provided in~\tabref{tab:gbaf_ablation}, demonstrate a surprising observation that the self-attention and PE modules, which are responsible for cross-symbol block coding, contribute only marginally to the performance of GBAF. This brings into question the value of block encoding and motivates us to find simpler designs that perform only symbol-by-symbol coding while providing competing performance. 

\begin{table}[!htp]
\centering
\begin{tabular}{ccccc}
\toprule
\textbf{SNR (dB)} & \textbf{GBAF} & \textbf{GBAF (no attn)} & \textbf{GBAF (no attn, no PE)} \\
\midrule
-1.5 & 9.8e-5 & 7.5e-5 & 2.6e-5 \\
-1.0 & 2.6e-9 & 1.2e-9 & 2.7e-9 \\
 ~0.0 & 5.4e-10 & 6.6e-10 & 8.1e-10 \\

\bottomrule
\end{tabular}
\captionsetup{font=small}
\caption{Effect of block coding on performance of GBAF for rate $\sfrac{3}{9}$ code with noiseless feedback. Positional encoding and self-attention modules have no noticeable effect on BLER performance.}
\label{tab:gbaf_ablation}
\vspace{-1em}
\end{table}

\subsection{Scaling laws: batch size}\label{sec:scaling}
In section \secref{sec:results}, we have demonstrated that it is possible to achieve performance comparable to GBAF with fewer parameters and lower compute complexity. We also hypothesized in \secref{sec:training} that it is important to train the encoder using a very large batch size to accurately capture the statistics of the distribution so that the available power can be optimally allocated to the symbols with error at the receiver while reducing the power to the remaining symbols in a batch. In deep learning literature, it is well known that increasing the batch size can noticeably improve the performance of the neural network model~\cite{bahri2021explaining}. To better understand the effect of batch size on \lightcode, we perform a systematic study by training \lightcode with different batch sizes. As noted in \figref{fig:deepcode_scaling}, at a batch size of $1.5\times10^3$, \lightcode has a BLER of $3.7\times10^{-9}$, similar to that of GBAF (w/o BU). However, the BLER drops to $4\times10^{-10}$, outperforming GBAF when the batch size increases to $5\times10^4$ and beyond. We would like to note that for each batch size, hyperparameters such as learning rate have been optimized, and the model is trained until saturation to ensure the best possible performance.

\begin{remark} \textnormal{ Using a lightweight network with a small number of parameters makes it feasible to train with a very large batch size, resulting in a significantly lower error rate. }

\begin{figure}[!htb]
    \centering
 	\includegraphics[width=\linewidth]{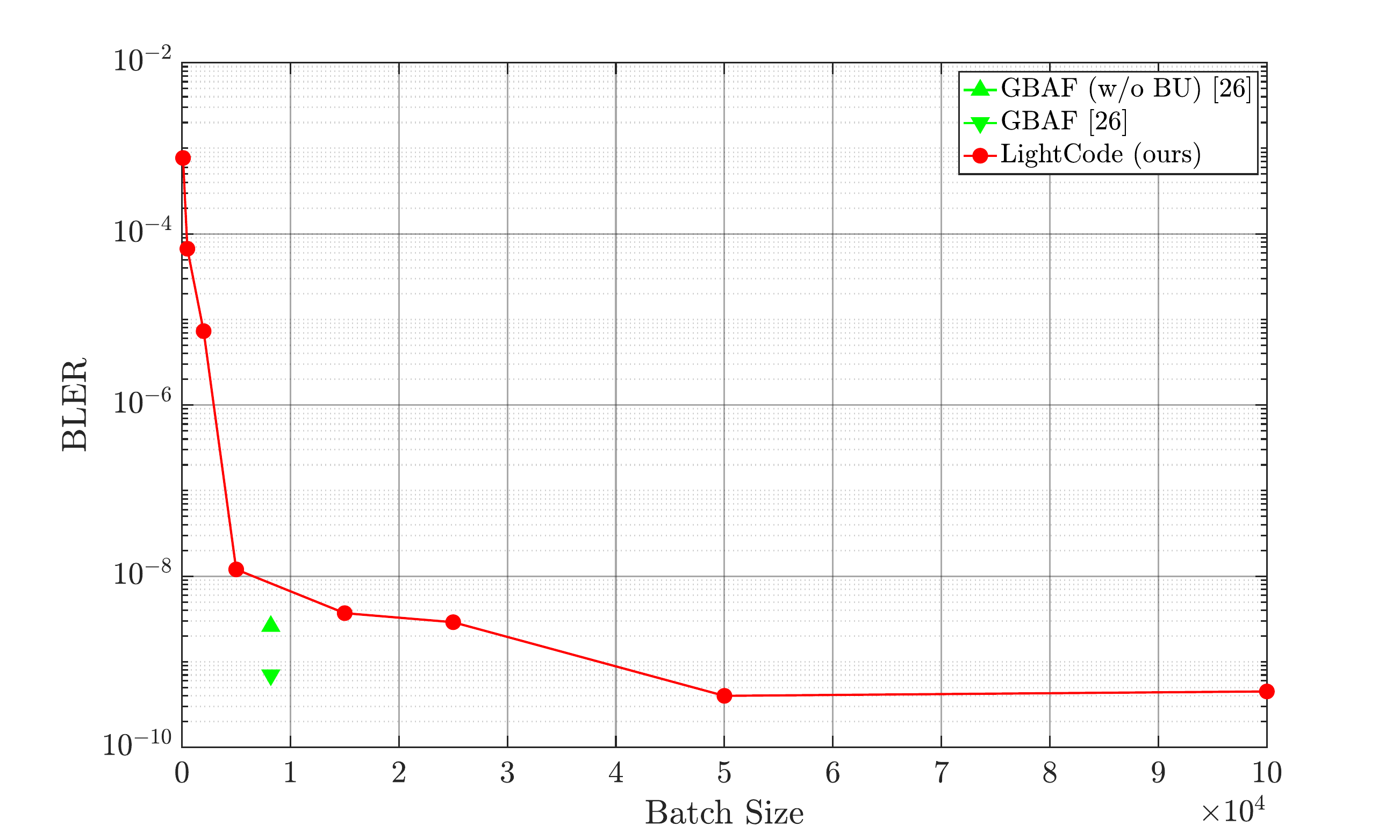}
 	\captionsetup{font=small}
 	\caption{BLER performance of \lightcode with with respect to training batch size for rate $\sfrac{3}{9}$ code at SNR -1.0 dB with noiseless feedback. Performance improves significantly with respect to batch size, surpassing the performance of GBAF at batch size $ > 5\times 10^4$.}
 	\label{fig:deepcode_scaling}
\end{figure}

\end{remark}

\subsection{Complexity and Throughput}\label{sec:comp_throughput}

We compare the total number of parameters in the encoder and decoder for a rate $\sfrac{3}{9}$ \lightcode against GBAF and Robust Coding within~\tabref{tab:num_params}. \lightcode reduces the parameter count by more than $10 \times$ compared to GBAF and Robust Coding, resulting in notable savings in memory.

\begin{table}[!ht]
\centering
\begin{tabular}{lc}
\hline
\textbf{Scheme} & \textbf{Total \# parameters} \\
\hline
GBAF~\cite{ozfatura2022all} & $9.1 \times 10^4$ \\
Robust Coding~\cite{kim2023robust} & $8.6 \times 10^4$ \\
\lightcode (ours) & $\mathbf{7.3 \times 10^3}$ \\

\hline
\end{tabular}
\captionsetup{font=small}
\caption{Rate $\sfrac{3}{9}$ \lightcode requires $< \sfrac{1}{10}^{\text{th}}$ the number of parameters compared to GBAF and Robust Coding.}
\label{tab:num_params}
\vspace{-1em}
\end{table}

\begin{table*}[!ht]
\centering

\begin{tabular}{lcccc}
\hline
\textbf{Scheme} & \textbf{ Enc (CPU)} & \textbf{ Dec (CPU)} & \textbf{ Enc (GPU)} & \textbf{ Dec (GPU)} \\
\hline
GBAF~\cite{ozfatura2022all} & $ 1.6 \times 10^5 $ & $ 1.7 \times 10^6 $ & $ 4.4 \times 10^8 $ & $ 3.3 \times 10^9 $\\

Robust Coding~\cite{kim2023robust} & $ 2.0 \times 10^5 $ & $  7.6 \times 10^4 $ & $ 4.1 \times 10^8 $ & $ 2.4 \times 10^9  $\\

\lightcode (ours) & $\mathbf{1.5 \times 10^6}$  & $\mathbf{1.3 \times 10^7 }$ & $\mathbf{2.9 \times 10^9}$  & $\mathbf{3.1 \times 10^{10} }$\\

\hline
\end{tabular}
\captionsetup{font=small}
\caption{Throughput (symbols/sec) comparison. Rate $\sfrac{3}{9}$ \lightcode achieves up to $ 10 \times $ higher decoding throughput compared to GBAF and up to $171 \times$ higher decoding throughput in CPU mode compared to Robust Coding.}
\label{tab:inference}
\end{table*}





To compare the throughputs, we measure the inference time. We define the encoder throughput $T_E$ as the number of message blocks or symbols encoded per second and the decoder throughput $T_D$ as the number of message blocks or symbols decoded per second. Formally, for a rate $\sfrac{K}{D}$ code, the throughput can be defined as 
\begin{align}
    T_{EK} &= (\sfrac{K}{D}) T_E \hspace{2em} (\text{bits/sec}),\\ 
    T_{DK} &= (\sfrac{K}{D}) T_D \hspace{2em} (\text{bits/sec}).
\end{align}
We measure the throughput in CPU mode on a \texttt{AMD Ryzen Threadripper PRO 5975WX 32-Cores} processor and in GPU mode using an \texttt{NVIDIA GeForce RTX 4090}, using a batch size of $10^5 - 10^7$. \lightcode provides up to $10 \times$ higher encoding and decoding throughput than GBAF. Further, \lightcode provides up to $171 \times$ higher decoding throughput in CPU mode and up to $10 \times$ higher throughput in GPU mode compared to Robust Coding, providing significant gains in latency, as shown in~\tabref{tab:inference}.

We note here that comparing throughputs for classical schemes against deep learning-based schemes is not straightforward and depends heavily on the implementation. Our implementation of \lightcode uses \texttt{PyTorch} libraries, whereas \pb and other classical schemes were implemented using \texttt{NumPy} libraries. Further, a significant computational bottleneck for the classical schemes is the necessity for demodulation after each round, the speed of which is heavily influenced by the implementation methodology. Hence, we restrict our comparison to the family of deep learning codes, where the latency primarily originates from the architecture complexity and the total number of parameters, and a fair comparison is feasible. 
However, it is clear that the total number of numerical operations in classical schemes is significantly lower than that of deep-learning-based schemes, and we acknowledge that the throughput of classical schemes can be significantly higher.

\medskip

\textbf{Complexity vs. BLER trade-off:} While analytically quantifying the complexity of deep learning-based channel coding schemes with respect to target BLER would be very interesting, it is highly non-trivial and a widely open problem. Thus, we instead conduct an empirical study to analyze how the complexity of \lightcode scales with target BLER. Empirically, we observed that reducing the complexity of the decoder while maintaining the same complexity for the encoder provided the best trade-off in performance vs complexity. In this experiment, we vary the target BLER and empirically find the required encoder-decoder dimensions for a rate $\sfrac{3}{9}$ code at a forward SNR of $-1.0$ dB and a noiseless feedback channel. The results for the same are shown in~\tabref{tab:complexity}, demonstrating an almost linear degradation of BLER in log scale with the number of parameters.

\begin{table}[H]
\centering
\begin{tabular}{lccc}
\hline
\textbf{BLER} & Enc dimension & Dec dimension & \textbf{Total \# params} \\
\hline
$4.5 \times 10^{-10}$ & 32 & 32 &$7.3 \times 10^3$ \\
$3.4 \times 10^{-9}$ & 32 & 16 &$4.7 \times 10^3$ \\
$5.2 \times 10^{-9}$ & 32 & 12 &$4.3 \times 10^3$ \\
$2.8 \times 10^{-8}$ & 32 & 8 &$3.9 \times 10^3$ \\

\hline
\end{tabular}
\captionsetup{font=small}
\caption{Complexity vs BLER trade-off for rate $\sfrac{3}{9}$ code at a forward SNR of $-1.0$ dB and noiseless feedback. BLER (in log scale) degrades almost linearly with a decrease in the number of parameters.}\looseness=-1
\label{tab:complexity}
\end{table}

\subsection{Interpretation of \lightcode}\label{sec:interpret}

By independently encoding sub-blocks of length $K=3$ and limiting to a symbol-by-symbol strategy, \lightcode allows for better interpretability and analysis of the learned encoder representations compared to block coding schemes such as GBAF. By analyzing the power allocation and relation between encoder output and feedback from previous rounds, we draw connections between \lightcode and \pb. 
\subsubsection{Power distribution}
A key contributing factor to the superior performance of \pb compared to SK is the discrete-symbol scheme in the final round of communication. In the high-SNR regime, we are only interested in the error in the PAM index of the decoded symbol with respect to the original symbol; this will result in a sparse distribution where most of the samples are 0. Thus, a majority of the available power is naturally allocated to the symbol locations with non-zero error. Surprisingly, we find similar behavior for \lightcode towards the final rounds of communication where the error is sparse. 

To test this hypothesis, we choose a moderately sparse error regime for ease of analysis. In~\figref{fig:power_dist}, we plot the power distribution of the encoder output in round $7$ for rate $\sfrac{3}{9}$ code at SNR $-1.0$ dB and noiseless feedback by randomly sampling $50$ symbols. On the X-axis, we see the sample number, and on the Y-axis, we plot the magnitude of error in the integer PAM index of the estimate after round $6$ and compare it against the magnitude of encoder output in round $7$. Note that a difference in index of 1 corresponds to a magnitude of 2 in the un-normalized PAM from Eqn.~\ref{eqn:pam}. It is evident from~\figref{fig:power_dist} that the highest power is allocated to the symbols with error in the estimate.


\begin{figure}[!htb]
    \centering
 	\includegraphics[width=\linewidth]{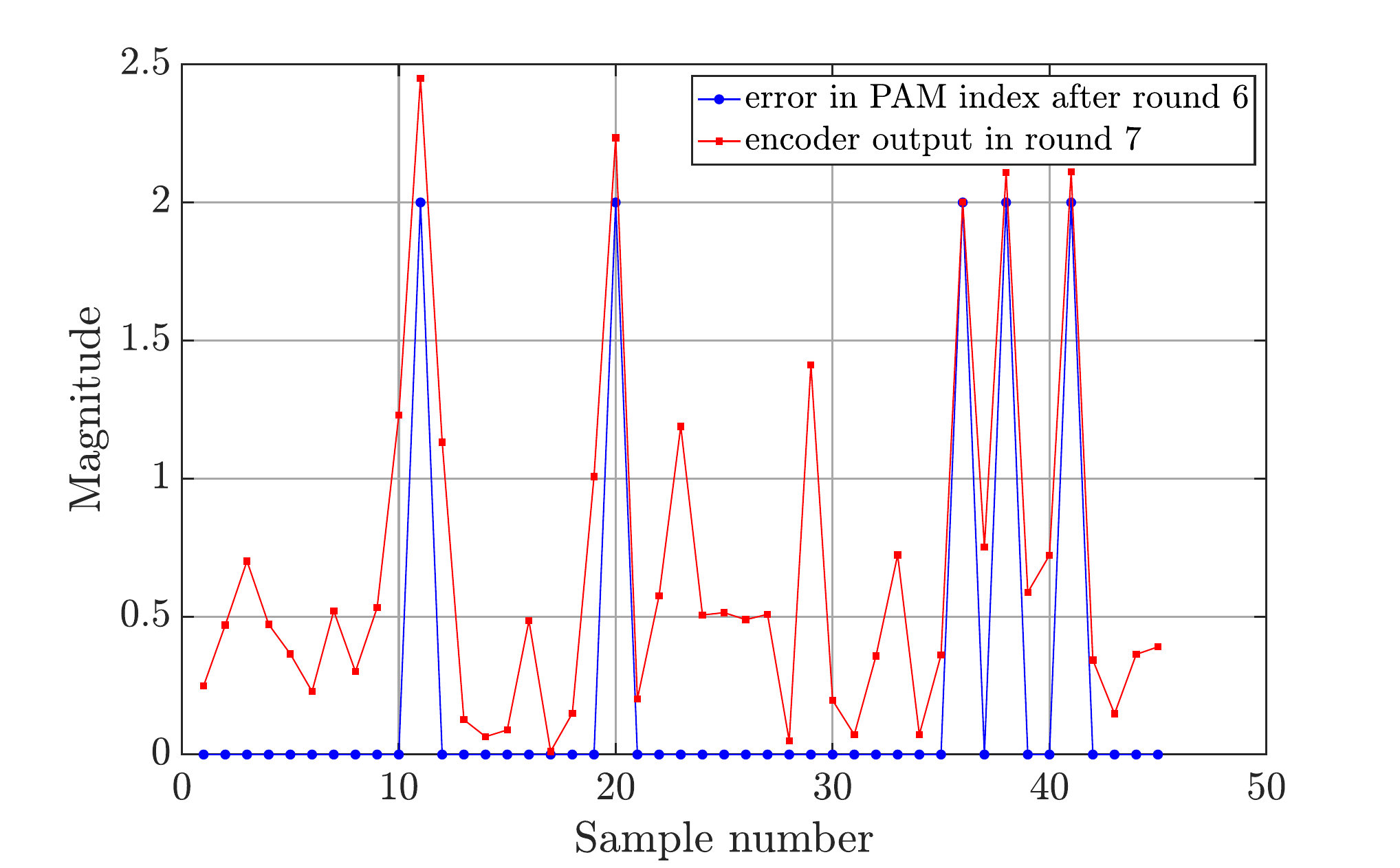}
 	\captionsetup{font=small}
 	\caption{\lightcode allocates more power to the symbols with error in estimate from the previous round, improving the overall probability of decoding.}
 	\label{fig:power_dist}
\end{figure}

\begin{figure}[!htb]
    \centering
 	\includegraphics[width=\linewidth]{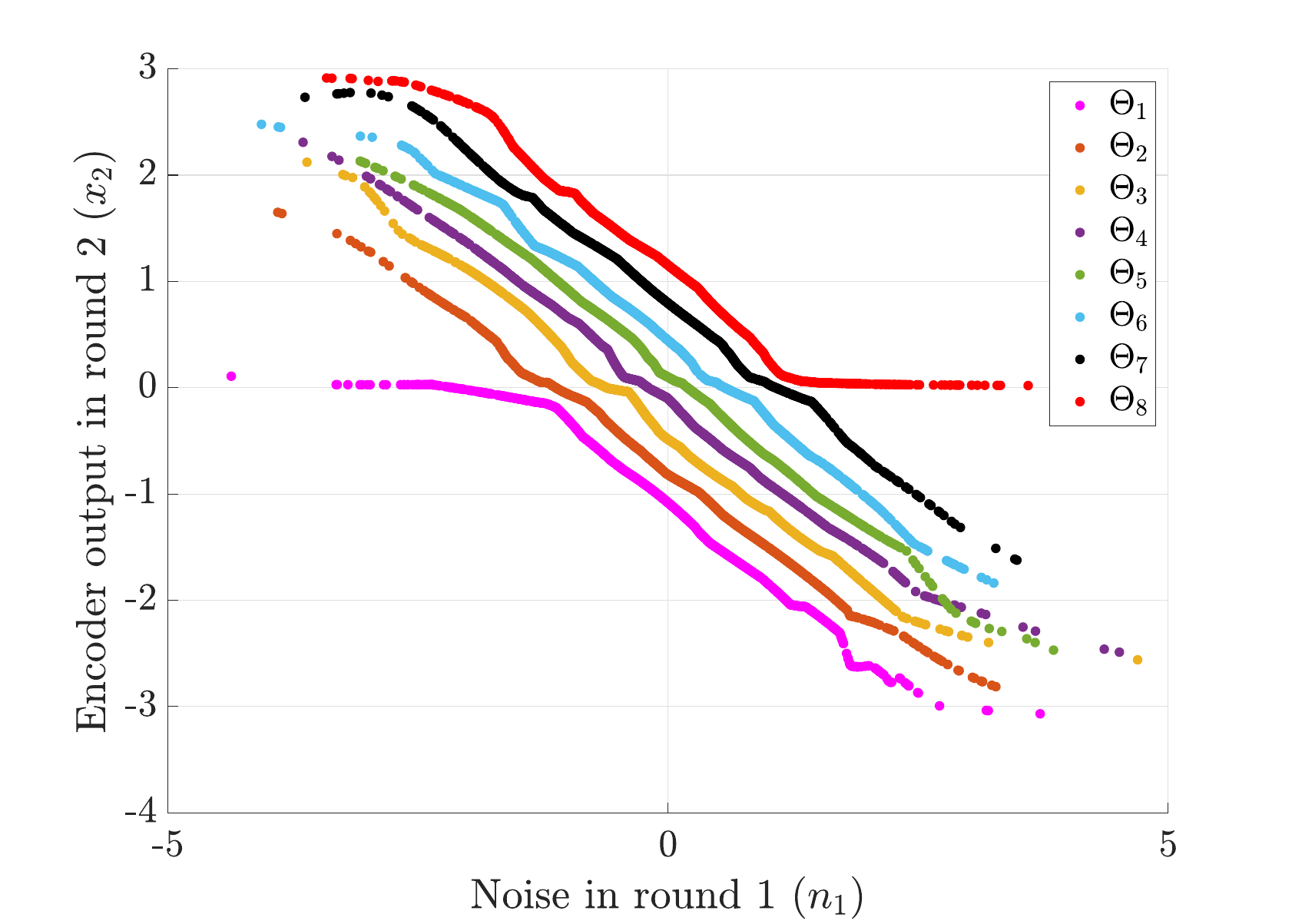}
 	\captionsetup{font=small}
 	\caption{Output of encoder in round 2 ($x_2$) shows that the encoder is directly compensating for the noise experienced in round 1 ($n_1$). }
 	\label{fig:n1_vs_p2}
\end{figure}

\begin{figure*}[!htb]
    \centering
 	\includegraphics[width=\linewidth]{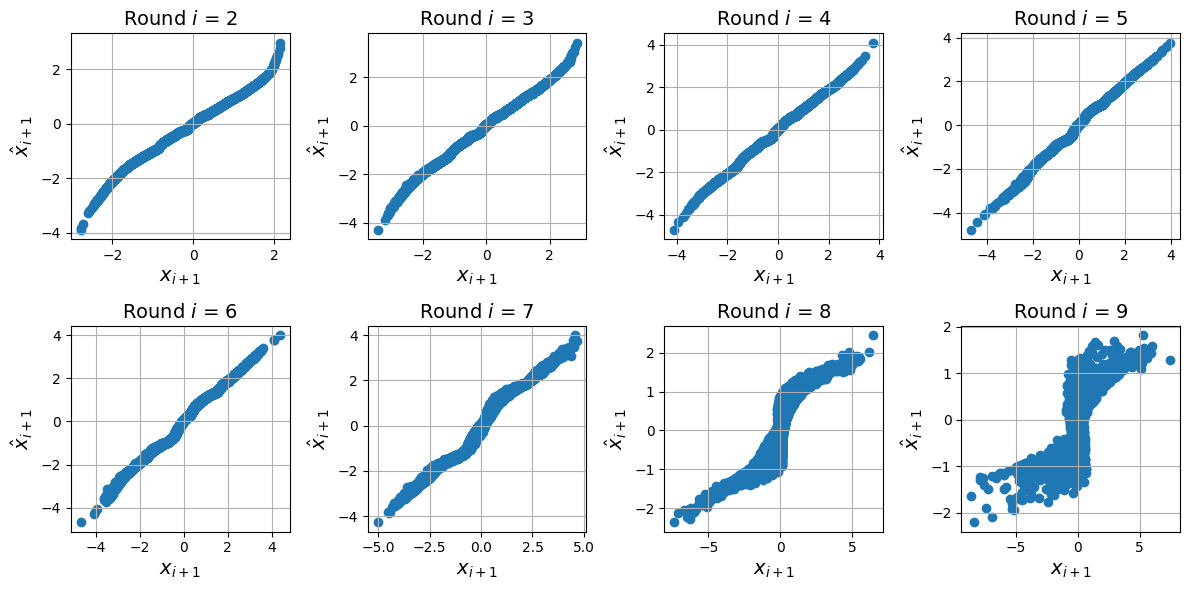}
 	\captionsetup{font=small}
 	\caption{Comparison of linear approximation (Y-axis) vs true encoder output (X-axis). As the rounds progress, the linearity of the relation between encoder output and the feedback from previous rounds breaks, making it harder for analytical schemes to perform well.}
 	\label{fig:linear_reg}
\end{figure*}

 \subsubsection{Interpreting the encoder} 
 In~\figref{fig:n1_vs_p2}, we plot the output of the encoder in round 2 with respect to the feedback in round 1. The encoder is approximately transmitting a linearly scaled version of noise from round 1. Interestingly, for the symbols on the boundary of the constellation embeddings, the encoder does not need to transmit any data when the noise favors the ground truth, unlike in the SK scheme, where all noise needs to be corrected. For instance, consider $\Theta_1$, mapped to the symbol at the boundary on the left, where a negative noise in round 1 is favored. And $\Theta_8$ is mapped to the symbol at the boundary on the right, where a positive noise in round 1 is favored. In such scenarios, the transmit power saved here can be reallocated to other symbols, improving the overall decoding error, which is only possible because of the non-linear activation functions.

Visualization of the encoder output beyond round 2 is difficult as the number of inputs to the encoder increases linearly with the number of rounds of communication. Alternatively, we test the linear dependency between $x_{i+1}$ and $\{x_1, \dots, x_i, n_1, \dots, n_i\}$ by performing a linear regression as
\begin{equation}
    \hat{x}_{i+1} = \sum_{j=1}^{i} \alpha_j x_j + \beta_j n_j + c,
\end{equation}
where $\alpha_j$ and $\beta_j$ are regression coefficients and $c$ is the intercept found using \texttt{LinearRegression} toolbox in \texttt{sklearn} over $10^6$ samples. To make the analysis tractable, we build one linear regression model for each symbol in the PAM constellation. In~\figref{fig:linear_reg}, we plot the predicted output using linear regression vs. true encoder output when the input PAM symbol index is $\Theta=1$, for rounds $j=2$ to $9$. It is evident that as the round number increases, the relation becomes highly non-linear, and hence, classical schemes such as SK and GN as other linear schemes~\cite{chance2011concatenated,mishra2023linear} fail to model these dependencies. 

We note that recent results in~\cite{devroye2022interpreting,zhou2024interpreting} show the possibility of using post-hoc interpretability techniques to analyze the deep-learning-based codes and learn analytical approximations for the learned codes. Making progress in this direction is an interesting avenue for future work. 

%% file: conclusion.tex
\section{Conclusion}\label{sec:conclusion}
In this work, we address the problem of designing lightweight coding schemes for channels with feedback. First, we propose an analytical scheme, \pb, that can be viewed as a combination of Schalkwijk-Kailath and \gn schemes. Using a hybrid strategy and taking advantage of the discrete nature of the signal in the final round, \pb noticeably outperforms both SK and GN, providing a performance that competes with current deep learning schemes in regions of high-SNR. 

Next, we propose a lightweight deep-learning-based scheme, \lightcode, that can achieve state-of-the-art BLER performance while using less than $\sfrac{1}{10}^\text{th}$ the parameters and compute complexity compared to existing deep-learning-based schemes. By limiting to a symbol-by-symbol strategy and carefully designing the feature extractor using skip connections, combined with a training strategy that uses a very large batch size of $10^5$, \lightcode achieves a BLER up to $\sim 10^{-10}$. 

Additionally, with the help of systematic ablation studies, we have demonstrated that the self-attention module in the transformer-based GBAF code has very little effect on the BLER performance, demonstrating the limited benefit of block coding in this regime. 

Further, we interpret the \lightcode to show that power distribution in the sparse error regime of \lightcode is similar to that of \pb, where a majority of the available power is allocated to the symbols with error. Finally, we also perform a linear regression on the encoder outputs and the feedback from previous rounds. Our findings show that although the early stages of communication exhibit a linear relationship, it becomes non-linear towards the later stages, underscoring the importance of employing deep-learning-based non-linear coding techniques to attain optimal performance in regions of extremely low SNR.

\section*{Acknowledgment}
This work was partly supported by ARO Award W911NF2310062, ONR Award N000142412542, NSF Award CNS-2008824, and the 6G$@$UT center within the Wireless Networking and Communications Group (WNCG) at the University of Texas at Austin.

\newpage

%% file: bibilography.tex
\vspace{-.8em} 
 \medskip
 \small
 \bibliographystyle{IEEEtran}
 \bibliography{bibilography}

\normalsize

%% file: appendix.tex
\begin{appendices}

\section{Longer blocklengths for \lightcode}\label{app:long_lightcode}

While the input and output dimensions for \lightcode can be increased to support learning longer blocklengths, converging to a good solution is harder because of the curse of dimensionality. In~\tabref{tab:bler_k_6_app}, we provide the performance of \lightcode trained for $K=6$ at different complexities of the encoder-decoder at forward SNR = $-1.0$ and noiseless feedback.

    \begin{table}[ht]
    \centering
    \begin{tabular}{lccc}
    \hline
    \textbf{BLER} & Enc dimension & Dec dimension & \textbf{Total \# params} \\
    \hline
    $3.7 \times 10^{-4}$ & 32 & 32& $7.3 \times 10^3$ \\
    $7.2 \times 10^{-5}$ & 128 & 64 & $1.09 \times 10^5$ \\
    $2.4 \times 10^{-6}$ & 128 & 128 & $4.1 \times 10^5$ \\
    \hline
    \end{tabular}
    \captionsetup{font=small}
    \caption{BLER for K=6 at forward SNR of $-1.0$ dB and noiseless feedback.}
    \label{tab:bler_k_6_app}
    \end{table}

We begin by comparing the performance of \lightcode for $K=6$ while maintaining the same encoder-decoder dimension used for $K=3$, which is $32$. From~\figref{fig:bler_3_9_noiseless}, the BLER of \lightcode trained with $K=3$ at a forward SNR of $-1.0$ dB and noiseless feedback is $4.5 \times 10^{-10}$. Hence, by transmitting $6$ bits as two sub-blocks, the BLER would be given by $1 - (1 - 4.5 \times 10^{-10})^2 = 9 \times 10^{-10}$. But, as seen from~\tabref{tab:bler_k_6_app}, the BLER performance for \lightcode trained with $K=6$ and a hidden dimension of $32$ is orders of magnitude higher, $3.7 \times 10^{-4}$. Next, we increase the hidden dimensions to $128$, increasing the total number of parameters by $100 \times$, and the performance is still much worse than \lightcode trained with $K=3$.

\newpage

\section{Encoder decoder architecture}\label{app:arch}
\begin{figure}[!h]
    \centering
    \includegraphics[width=0.9\linewidth]{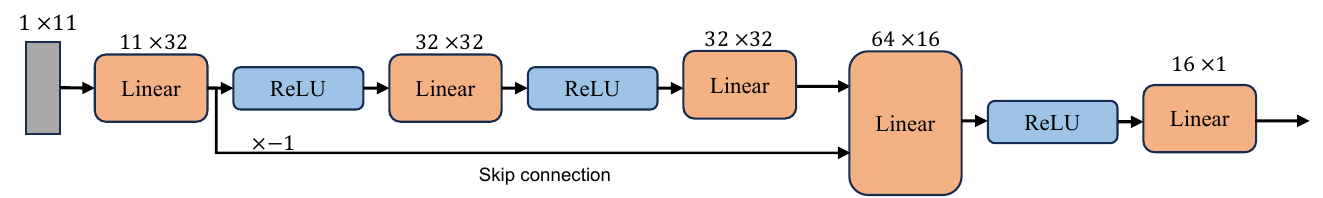}
    \captionsetup{font=small}
    \caption{Encoder}
    \label{fig:encoder}
    \vspace{0.5em}  

    \includegraphics[width=\linewidth]{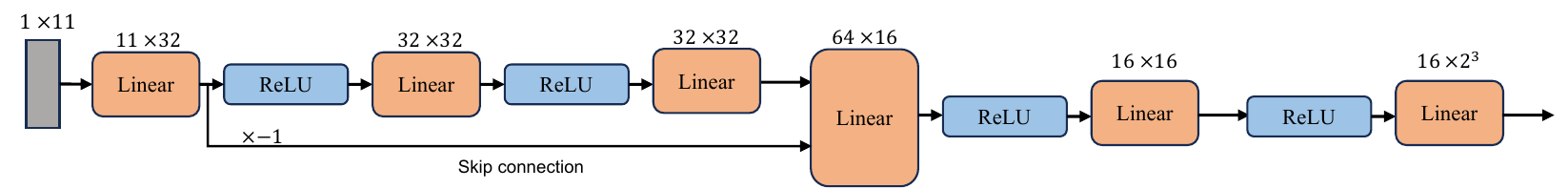}
    \captionsetup{font=small}
    \caption{Decoder}
    \label{fig:decoder}
    
    \captionsetup{font=small}
    \caption{Full architecture for \lightcode}
    \label{fig:encoder_decoder_arch}
\end{figure}




\section{Algorithm for \gn}\label{app:alg_gn}

\begin{algorithm}[hbt!]
\caption{\gn coding scheme}
\label{alg:GN}

\KwInput{Message symbol $\Theta$, number of rounds $D$, noise variance $\sigma_{ff}^2$ }

\vspace{5pt}
\textbf{Round 1:}\
\textbf{Tx:} 
Power normalization: $X_1 = \sqrt{P_1}\Theta$\;
Transmit: $Y_1 = X_1 + N_1$\;
\textbf{Rx:} Send $Y_1$ as feedback to Tx \;


\textbf{Round 2:}\
\textbf{Tx:} 
Power normalization: $X_2 =  \frac{\sqrt{P_2}U_2}{\sigma_2}$; $U_2 = N_1$ and $\sigma_2 = \sigma_{ff}$\;
Transmit: $Y_2 = X_2 + N_2$\; 
\textbf{Rx:} Compute the 
LMMSE estimate of transmit symbol $\E [U_2 | Y_2] = \frac{\sigma_2 \sqrt{P_2}Y_2}{1 + P_2}$\;
\vspace{5pt}
Update the estimate as $\hat{X_1} = X_1 - \E [U_2 | Y_2]$

\vspace{5pt}
\texttt{/* Tx sends the error in estimate $\epsilon_2 = \E [U_2 | Y_2] - U_2$ over $D-2$ rounds */} \\
\vspace{5pt}

\While{$3 \leq i \leq D-1$}{  
    \textbf{Tx:} Compute the error in estimate of previous round  $U_{i} = \E [U_{i-1} | Y_{i-1}] - U_{i-1}$ \;
    Power normalization: $X_i =  \frac{\sqrt{P_2}U_i}{\sigma_i}$, $\sigma^2_{i} = \frac{\sigma^2_{i-1}}{1+S_{i-1}}$\;
    Transmit:  $Y_{i} = X_{i} + N_{i}$\;
    \textbf{Rx:} Send feedback as LMMSE estimate of transmit symbol $\E [U_i | Y_i] = \frac{\sigma_i \sqrt{P_2}Y_i}{1 + P_2}$ \;
    \vspace{5pt}
    Update the estimate as $\hat{X_1} = X_1 - \sum_{j=2}^{j=i} \E [U_j | Y_j]$
}
\vspace{5pt}
\textbf{Decoding after round $D-1$}: Map $\hat{\Theta}_{D-1} = \frac{\hat{X_1}}{\sqrt{P_1}\eta}$ to the closest symbol in the $2^K$ PAM constellation.

\vspace{5pt}
\texttt{/* High-SNR scheme for rounds $D$ */} \\
\vspace{5pt}

\textbf{Round D:}\
\textbf{Tx:} Compute the difference in PAM indices $U = \hat{M} - M$, where $\hat{M}$ and $M$ correspond to the integer index from PAM constellation for $\hat{\Theta}$ and $\Theta$ respectively\;
Transmit: $Y_D = X_D + N_D$\;
\vspace{5pt}
\textbf{Final Decoding}: Use ML decoder to detect $\hat{U}$ and detect original PAM signal using $\hat{U}$.

\end{algorithm}

\end{appendices}